\documentclass[journal]{IEEEtran}
\usepackage{amssymb}

\usepackage{lipsum}
\usepackage[english]{babel}
\usepackage{mathrsfs}
\newlength{\halfpagewidth}
\divide\halfpagewidth by 2

\newtheorem{corollary}{Corollary}

\newtheorem{theorem}{Theorem}
\newtheorem{lemma}{Lemma}
\newtheorem{remark}{Remark}

\newenvironment{proof}{{\indent \it {Proof}:\quad}}{\hfill $\square$\par}
\newenvironment{corollarybox} {\begin{corollary}}{\hfill \interlinepenalty500 $\Box$\end{corollary}}
\newenvironment{theorembox} {\begin{theorem}}{\hfill \interlinepenalty500 $\Box$\end{theorem}}

\newtheorem{innerdefinition}{Definition}
\newenvironment{definition}
  {\begin{innerdefinition}}
  {\hfill $\square$\end{innerdefinition}}

\newenvironment{remarkbox} {\begin{remark}}{\hfill \interlinepenalty500 $\Box$\end{remark}}

\usepackage{amsmath}
\usepackage[colorlinks=true, allcolors=blue]{hyperref}
\usepackage{color}
\usepackage{lettrine}
\usepackage{amssymb}
\usepackage{subcaption}
\usepackage{booktabs}
\usepackage{graphicx}

\usepackage{tabularray}
\usepackage{array}
\title{Fluid Antenna Enabled Compact Ultra Massive Antenna Array for Satellite Communications}
\author{Tianyu Han, Yongxu Zhu, Gan Zheng, Pantelis-Daniel Arapoglou
\vspace{-8mm}

\thanks{T. Han and G. Zheng are with the School of Engineering, University of Warwick, Coventry, UK (Email: {tianyu.han, gan.zheng}@warwick.ac.uk).}
\thanks{Y. Zhu is with the National Communications Research Laboratory, Southeast University, Nanjing, China (Email: yongxu.zhu@seu.edu.cn).}
\thanks{Pantelis-Daniel Arapoglou is with the European Space Agency Research
and Technology Centre (ESA/ESTEC), 2201 AZ Noordwijk, The Netherlands
(e-mail: pantelis-daniel.arapoglou@esa.int).}
}

\begin{document}
\maketitle

\begin{abstract}
Satellites provide seamless coverage and are critical for emergency communications during natural disasters. However, their performance is constrained by limited spectrum and high deployment cost. To address these issues, we propose a fluid antenna system (FAS)-based solution that enables dynamic signal adaptation. Building on this concept, a compact ultra-massive antenna array (CUMA) is introduced, where multiple ports are simultaneously activated to coherently combine signal components. This design mitigates interference while reducing cost, as each fluid antenna requires only a single RF chain yet achieves significant improvement in the received signal-to-interference-plus-noise ratio (SINR). We consider a satellite CUMA network where all ground users share the same satellite for uplink transmission, and CUMA is employed to suppress inter-user interference. Closed-form expressions for the received signal power, interference power, and their distributions are derived. Based on these results, the outage probability is obtained in a unified form along with an accurate approximation, and the ergodic rate is characterized. Our analysis identifies the conditions under which CUMA outperforms maximum ratio combining in satellite systems. Notably, with sufficiently compact fluid antenna configurations, the received signal becomes deterministic, indicating that system performance is dominated by interference statistics. Moreover, increasing the number of ports yields a linear beamforming gain. Numerical results further compare orthogonal and non-orthogonal multiple access CUMA, showing that the latter achieves superior performance under wideband conditions.

\end{abstract}

\begin{IEEEkeywords}
Satellite communications, fluid antenna, compact ultra massive antenna array.
\end{IEEEkeywords}

\section{Introduction}

\IEEEPARstart{S}{atellite} communication plays a pivotal role in addressing the escalating demand for worldwide connectivity, which is envisioned as one of the key goals of 6G \cite{Satellite_6G}. Despite offering extensive coverage, the deployment of satellite systems incurs substantial costs, primarily due to the weight of the satellite. Moreover, satellite systems face constraints in spectral allocation, which vary across countries and regions. The spectrum assigned to satellite communications is limited, making effective resource management and allocation a complex challenge. This spectrum scarcity, combined with the variability in frequency bands allowed in different jurisdictions, often results in users sharing the same frequency being geographically close to each other. Such interference can significantly degrade communication quality and reliability, presenting a unique challenge to satellite communications networks. 

\subsection{Satellite, CUMA and Literature Review}

Motivated by the aforementioned discussion, it is crucial to explore methods that enhance satellite communications resilience to interference while maintaining a lightweight system design to meet deployment constraints. Recent advancements in satellite communications increasingly integrate multiple-input multiple-output (MIMO) technology. However, the predominance of the line-of-sight (LoS) propagation path, coupled with the high spatial correlation induced by rain attenuation, presents significant challenges to achieving independent channels and full-rank MIMO performance. To address these limitations, satellite MIMO systems often necessitate the deployment of separate satellites \cite{MIMO_Satellite}. It is important to note, however, that independent channels are essential for multiplexing gain but not for beamforming gain \cite{tse2005fundamental}. An alternative approach proposes leveraging different polarization to achieve channel independence without significant separation between satellites. Early works confirmed that dual polarization can enhance the performance of a single satellite by utilizing polarization diversity \cite{Polarization_Single_Satllite,Polarization_Single_Satllite_2}. Lately, the authors of \cite{Polarization_Loss} investigated the effects of polarization conversion loss on the performance of non-orthogonal multiple access (NOMA) satellite networks.

Recently, a novel, lightweight reconfigurable antenna technology, known as fluid antenna system (FAS) \cite{Fluid_antenna_system}, has been proposed as an effective solution to mitigate interference in various communication scenarios. This technology leverages the mobility of the antenna, thereby reducing the number of RF chains, which are not only expensive but also heavy. FAS enables multiple users to be served on the same time-frequency resource, a concept referred to as fluid antenna multiple access (FAMA), where FAS could select the "best" port among all the ports, thereby improving the received signal-to-interference ratio (SIR). In \cite{f_FAMA}, it was suggested that FAS port can be switched at the symbol level, a method termed fast FAMA ($f$-FAMA). Later, \cite{s_FAMA} introduced slow FAMA ($s$-FAMA), where the FAS ports are switched only when the fading channel changed. Both the results in \cite{f_FAMA} and \cite{s_FAMA} confirmed that under the rich scattering scenario, $f$-FAMA and $s$-FAMA is a feasible solution to combat interference. Recently, another version of FAMA, referred to as compact ultra massive antenna (CUMA), was proposed in \cite{CUMA}. In CUMA, instead of activating just the best port, a large number of ports with aligned channels are activated to enhance the received signal for detection. The standard version of CUMA, as proposed in \cite{CUMA}, requires two RF chains: one dedicated to processing the in-phase component, selecting the larger positive or negative values from the set of activated ports, and the other to the quadrature component, applying the same selection criterion. Subsequently, in \cite{CUMA_4}, the authors extended CUMA by employing four RF chains, utilizing both the positive and negative parts of the in-phase and quadrature components for detection. Further, \cite{CUMA_even} generalized these results by employing any even number of RF chains, thereby improving performance.

\subsection{Aim and Contributions}

The limited spectrum and inevitable interference problem \cite{satellite_interference}, coupled with high launch costs \cite{satellite_launch_cost}, hinders the widespread adoption of satellite communications. In this paper, we propose equipping satellite with FAS and employing CUMA, which provides a lightweight and effective approach to combating interference. While the results in \cite{CUMA} indicate that the full benefits of CUMA are realized in a scattering environment, they also demonstrate that CUMA remains a viable solution even when only limited scattering is available. Despite the valuable insights offered by the previous research in CUMA, its performance when only a LoS path exists remains unexplored. The LoS propagation scenario renders the use of correlation structures in CUMA analysis infeasible. However, it also raises the prospect that ports with specific signs of in-phase and quadrature components can be identified in a compact form, thereby making the analysis feasible. Built upon this, this paper examines the performance of CUMA in a satellite communications scenario, specifically focusing on the case where only a LoS path is present. We have made the following detailed contributions:
\begin{itemize}
\item We consider the LoS propagation scenario in satellite communications, where the target user and interfering users and  are located in close proximity, and distinguishing them based on azimuth and elevation angles becomes particularly challenging. Exact and compact expressions for both the received CUMA signal and interference power are derived, enhancing the analytical tractability of the system. Using these compact expressions, we further derive the exact probability density function (PDF) for the received CUMA signal power and an approximate PDF for the overall interference power.
\item Based on the distributions of the received CUMA signal and interference power, we derive the PDF expression for the received signal-to-interference-plus-noise ratio (SINR), which is presented in an integral-form. Additionally, we provide a simplified closed-form PDF expression for the SINR, applicable for the special case where a sufficiently compact fluid antenna is adopted.
\item Using the derived PDF, we present the exact outage probability expression for the typical user, formulated in a single integral-form, along with an approximate closed-form expression. Based on these results, we derive the ergodic rate for the satellite CUMA network. Additionally, we assess the beamforming gain provided by CUMA in satellite communications and establish the condition under which CUMA outperforms maximum ratio combining (MRC) in the compact fluid antenna setting. Our analysis further demonstrates that, in the LoS communication scenario, a single RF chain is sufficient for CUMA  to support each user.
\item Our simulation results confirm the accuracy of all derivations and demonstrate that, although the derivations are based on the assumption that the density $\mu$ is an even integer, the proposed PDF remains valid even in the compact fluid antenna setting. Moreover, increasing the fluid antenna port density $\mu$ can effectively reduce the outage probability of a typical user. For scenarios with larger bandwidth, non-orthogonal CUMA ($N$-CUMA) proves to be a more effective solution compared to orthogonal CUMA ($O$-CUMA).
\end{itemize}

It is worth noting that the considered LoS scenario is not limited to satellite communications. As the frequency employed in communication systems increases, particularly in millimeter-wave bands and beyond, propagation paths may become predominantly or even exclusively LoS. Furthermore, the compact user setup aligns with the growing trend of smart device proliferation. This highly compact user setting represents a worst-case scenario in practical wireless communication, where the target and interfering users are spatially very close. The structure of this paper is organized as follows. Section \ref{sec:model} provides a detailed introduction to the satellite CUMA network model. In Section \ref{sec:analysis}, the primary analytical results are presented, encompassing the expressions and distribution of CUMA signal power, interference power, and SINR, alongside the outage probability and ergodic rate of the satellite CUMA network. Section \ref{Sec:Physical_Insights} offers some physical insights, with a particular focus on the beamforming gain achieved through CUMA. In Section \ref{sec:results}, numerical results are provided and finally, we provide some concluding remarks in Section \ref{sec:conclude}.

\section{System Model}\label{sec:model}
We consider a communication satellite equipped with $U$ fluid antennas, each contains a single RF chain and $K$ ports evenly distributed along a linear dimension of length, $W\lambda$, where $W$ is the scaling factor and $\lambda$ is the carrier wavelength. Consequently, the density of fluid antenna port, measured per $\lambda$, is denoted as $ \mu =\frac{K-1}{W}$ \footnote{In this paper, $W$ is assumed to be an integer, and $\mu$ is taken as an even integer to avoid non-integer issues in the port index, thereby facilitating analytical convenience. In practice, these configurations are determined by industry standards; thus, the insights provided in this paper can serve as guidance for industry design. To facilitate the analysis, the mutual coupling between ports is neglected in this work. A detailed discussion of mutual coupling in fluid antenna system can be found in \cite[Section III.E]{CUMA}}. The first port on the fluid antenna, used for CSI detection and determining the initial phase of the received signal, serves as the reference port and cannot be activated for reception. The rest of $K-1$ ports can be activated simultaneously. Each fluid antenna is assigned to decode signal from a single ground user, with each ground user equipped with a fixed-position antenna for transmitting the uplink signal. The same time-frequency resource was shared among $U$ users. As a result, each fluid antenna receive the one desired signal and $U-1$ interfering signals. This non-orthogonal multiple access CUMA scheme is referred to as $N$-CUMA.
Without loss of generality, this paper analyzes the outage probability of a typical user, as the model is applicable to all users. Additionally, the ergodic rate of the $N$-CUMA network is presented. Fig.~\ref{fig: MISO_FAS_Network} depicts our model. 
 
\begin{figure*}[]
\centering
\includegraphics[width=0.85\linewidth]{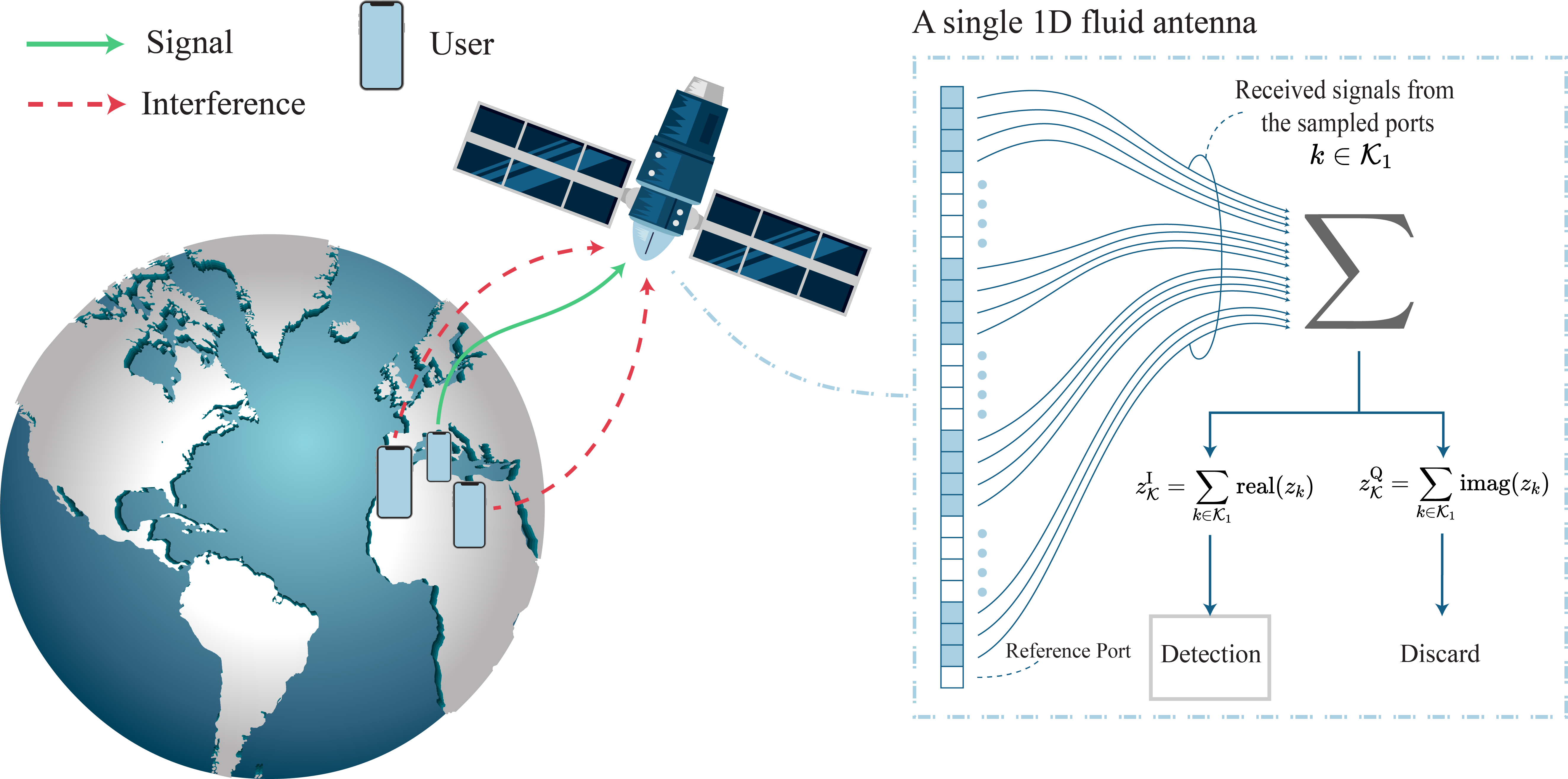}
\caption{A conceptual uplink satellite communications system, where the satellite employs multiple 1D fluid antennas to serve multiple ground users. Multiple access is facilitated using CUMA at the satellite, with a massive number of ports on the fluid antenna sharing a single RF chain to receive and aggregate signals for detection.}\label{fig: MISO_FAS_Network}
\end{figure*}

\subsection{Satellite Fluid Antenna Channel Model}
Due to the propagation characteristics of ground-to-satellite communication, particularly the strong LoS component, the fading can be effectively modeled using the free-space path loss model\footnote{Other large-scale fading coefficients, such as rain attenuation, could also be incorporated into the channel model that expressed in $\zeta_{u}$. Due to the dense user distribution, rain attenuation can be considered identical across all users and is therefore omitted here for brevity.}, expressed as
\begin{align}\label{Path_Loss}
    h_{u,k} &= \zeta_{u,k}^{\frac{1}{2}}e^{j\left(\psi_{u}+\frac{2\pi}{\lambda}r_{u,k}\right)} \nonumber \\
    &\overset{(a)}{\approx} \zeta_{u}^{\frac{1}{2}} e^{j\left(\psi_{u}+\frac{2\pi}{\lambda}r_{u,k}\right)}\nonumber \\ 
    &\overset{(b)}{=} \zeta_{u}^{\frac{1}{2}} e^{j\psi_{u}+j2\pi\frac{\upsilon_{u}}{\mu}(k-1)},
\end{align}
where $r_{u,k}$ is the distance between the $u$-th ground user and the $k$-th port of fluid antenna onboard the satellite, and $\zeta_{u,k}^{\frac{1}{2}} = \frac{\lambda}{4\pi r_{u,k}}$ is the corresponding free space path-loss. $\psi_{u}$ is the phase of the received signal at the reference port of fluid antenna, which follows a continuous uniform distribution in $(0,2\pi)$, denoted as $\mathcal{U}(0,2\pi)$. In $(a)$, we approximated  $\zeta_{u}  = \overline{\zeta_{u,k}}$, due to the fact the distance difference between each port is small enough compared to the distance between satellite and ground user and therefore can be ignored. In other words, only the phase difference can be observed on the different fluid antenna ports on the satellite. In $(b)$, $\frac{\upsilon_{u}}{\mu}(k-1)$ is the propagation distance between the $k$-th port and reference port from ground user $u$, normalized by wavelength $\lambda$. And $\upsilon_{u} = \sin{\theta_{u}}\cos{\phi_{u}}$, where $\theta_{u}$ and $\phi_{u}$ are the azimuth and elevation angle-of-arrival of the corresponding received signal from ground user $u$, receptively. This channel model is a modified version of the one presented in \cite{CUMA}, adapted for LoS propagation.

\subsection{Signal Model}
Under the aforementioned propagation model, the received signal at the $k$-th port of fluid antenna is given by
\begin{align}\label{signal_with_interference}
z_{k} = \sqrt{PG}h_{u,k}s_{u} +  \sqrt{PG}\sum_{\tilde{u}=1\atop \tilde {u}\ne u}^{U}h_{\tilde{u},k}s_{\tilde{u}} + \eta_{k},
\end{align}
where $P$ and $G$ represent the transmitter power and overall antenna gain, respectively, both assumed to be identical across all users. $s_{u}$ denotes signal symbol from the desired ground user, $s_{\tilde{u}}$ is the symbol intended from the $\tilde{u}$-th interfering ground user. $\eta_k$ denotes the zero-mean complex additive white Gaussian noise (AWGN) received at the $k$-th receive port, with a power of $\sigma_\eta^2 = N_0 B$, where $B$ represents the system bandwidth and $N_0 = c_{\mathrm{B}} T$ is the single-sided noise power spectral density. Here, $c_{\mathrm{B}}$ denotes Boltzmann's constant and $T$ is the receiver temperature. The power of the information symbol is ${\rm E}[|s_{u}]^2]= \sigma_{s}^2$.

In the proposed system, we consider a real-valued modulation scheme, where $s_{u} = s_{u}^{\mathrm I}+js_{u}^{\mathrm Q} = s_{u}^{\mathrm I}$, the quadrature (imaginary) part of the information symbol is zero. Consequently, the in-phase and quadrature components of the received signal at the $k$-th port of fluid antenna are expressed as follows
\begin{align}\label{I_signal_port_k}
z_{k}^{\mathrm{ I}}&={{\mathsf{real}}\left (h_{u,k}\right)}s_{u}^{\mathrm{ I}} +{\mathrm{real}}\left ({{\sum _{\substack{\tilde {u}=1\\ \tilde {u}\ne u}}^{U}h_{\tilde {u},k}s_{\tilde {u}}+ {\eta }_{k}}}\right), 
\\ \label{Q_signal_port_k}
z_{k}^{\mathrm{ Q}}&={{\mathsf{imag}}\left (h_{u,k}\right)}s_{u}^{\mathrm{I}} +{\mathsf{imag}}\left ({{\sum _{\substack{\tilde {u}=1\\ \tilde {u}\ne u}}^{U}h_{\tilde {u},k}s_{\tilde {u}}+ {\eta }_{k}}}\right). \end{align}

\subsection{CUMA}

Different from the current FAMA technology, CUMA enhances the received signal by  selectively activating $\overline{K}$ ports from the $K$ available ports. To simplify the notation, we define $\overline{K} = |\mathcal{K}|$, where $\mathcal{K}$ represents the set of indices corresponding to the activated ports. Given the activated ports set $\mathcal{ K}$, the received signal can be expressed as 
\begin{align} \begin{cases} \label{aggregated_signal_I_Q}
\displaystyle {z}_{\mathcal{ K}}^{\mathrm{I}}&=\sum _{k\in {\mathcal{K}}}{\mathsf{real}}\left ({z_{k}}\right),\\ \displaystyle {z}^{\mathrm{ Q}}_{\mathcal{K}}&=\sum _{k\in {\mathcal{K}}}{\mathsf{imag}}\left ({z_{k}}\right), \end{cases}\end{align}
where ${z}_{\mathcal{ K}}^{\mathrm{I}}$ and ${z}_{\mathcal{ K}}^{\mathrm{Q}}$ are the in-phase and quadrature components of the aggregated signal respectively. Substituting \eqref{I_signal_port_k} and \eqref{Q_signal_port_k} into \eqref{aggregated_signal_I_Q}, we have 
\begin{align} 
z^{\mathrm{I}}_{\mathcal{K}}&= s_{u}^{\mathrm{ I}}{\sum_{k\in {\mathcal{K}}}{\mathsf{real}}\left (h_{u,k}\right)} +\sum_{k\in {\mathcal{ K}}}{\mathsf{real}}\left ({{\sum _{\substack{\tilde {u}=1\\ \tilde {u}\ne u}}^{U}h_{\tilde {u},k}s_{\tilde {u}}}}\right) \nonumber \\  &\quad + {\sum _{k\in {\mathcal{ K}}}{\mathsf{real}}\left ({\eta }_{k}\right)}, \\
z^{\mathrm{Q}}_{\mathcal{ K}}&= s_{u}^{\mathrm{I}}{\sum _{k\in {\mathcal{ K}}}{\mathsf{imag}}\left (h_{u,k}\right)} +\sum _{k\in {\mathcal{ K}}}{\mathsf{imag}}\left ({{\sum _{\substack{\tilde {u}=1\\ \tilde {u}\ne u}}^{U}h_{\tilde {u},k}s_{\tilde {u}}}}\right) \nonumber \\
&\quad + {\sum_{k\in {\mathcal{ K}}}{\mathsf{imag}}\left ({\eta }_{k}\right)}.
\end{align}
Note that the aforementioned process only takes a single RF chain, as it involves merely the superposition of the in-phase and quadrature components of the received signal across the activated ports.  

It is evident that the performance of CUMA is dominated by the activated ports in $\mathcal{K}$. Different from previous CUMA work in \cite{CUMA}, where two sets were suggested, here, we consider a single set $\mathcal{K}_{1}$, which contains those ports with positive in-phase channels for the desired signal i.e, $\mathcal{K}_{1} = \left\{ k\mid k \in \{2,3, \dots, K\}~\text{and}~\mathsf{real}\left(h_{u,k}\right) > 0 \right\}$. Therefore, $z^{\mathrm{I}}_{\mathcal{K}_{1}}$ is the superposition over all the positive in-phase channels for the desired signal, while interference and noise are superposed randomly. And $z^{\mathrm{Q}}_{\mathcal{K}_{1}}$ is too noisy for useful reception and thus discarded. As a result, the received in-phase SINR can be enhanced. Later in Section \ref{Sec:Physical_Insights}, we demonstrate that the use of other set yields almost identical result as $\mathcal{K}_{1}$. For notation simplify, the index $\mathcal{K}_{1}$ is omitted in the following discussion.

Subsequently, the received CUMA in-phase SINR can be expressed as
\begin{align}\label{SINR_In_Phase_K_1}
     {\mathrm{ SINR}}_{\mathrm{I}}=\frac {\alpha _{\mathrm{ I}}}{\beta _{\mathrm{ I}}+\frac{\bar{K}}{2\Gamma}},
\end{align}
where 
\begin{align}\label{CUMA_signal}
    \alpha _{\mathrm{ I}} &\triangleq \left[\sum_{k\in {\mathcal{K}}_{1}}{\mathsf{real}}\left (h_{ {u},k}\right)\right]^2, \\ 
    \label{CUMA_interference_u}
    Y_{\tilde{u}} &\triangleq \left[\sum _{k\in {\mathcal{ K}}_{1}}{\mathsf{real}}\left (h_{\tilde {u},k}\right)\right]^2, \\
    \label{CUMA_interference_Slow}
    \beta_{\mathrm{I}} &\triangleq \sum_{\substack{\tilde {u}=1\\ \tilde {u}\ne u}}^{U}Y_{\tilde{u}}, \\
    \label{SNR}
    \Gamma &\triangleq PG\dfrac{\sigma_{s}^{2}}{\sigma_{\eta}^{2}}.
\end{align}
$\alpha_{\mathrm{I}}$ represents the signal power normalized with respect to antenna gain and transmit power. $Y_{\tilde{u}}$ denotes the normalized interference power from user $\tilde{u}$. $\beta_{\mathrm{I}}$ is the normalized total interference power collected on the in-phase components. $\Gamma$ is the average signal-to-noise ratio (SNR) without taking into account the free-space path loss. For simplicity, the term normalized is omitted hereafter.

\section{PERFORMANCE ANALYSIS}\label{sec:analysis}
In this section, we present the primary results evaluating the performance of the satellite CUMA system. To begin with, we express the CUMA signal in a compact form, as well as the interference, in the scenario where users sharing the same time-frequency resources are densely distributed. This scenario is particularly relevant, as satellite communications frequency are allocated differently across countries and regions; thus, users sharing the same frequency are typically confined to relatively close proximity. Subsequently, the PDFs of both the CUMA signal and the interference are provided, along with an analysis of their independence. As a result, the SINR expression of satellite CUMA is derived, with its simplified form under the assumption of compact fluid antenna. Finally, the performance metrics of satellite CUMA are evaluated, including the outage probability of a typical user and the ergodic rate of the network.

Before proceeding, we would like to clarify the assumptions made during the derivation. As a result of the dense user distribution and the large propagation distance, the azimuth angle $\theta_{u}$ and elevation angle $\phi_{u}$ of arrival for signals from different users are assumed to be identical, and therefore, $\upsilon_u = \upsilon_{\tilde{u}}.$ Moreover, to fully activate the potential of LoS CUMA, we consider that the fluid antenna on the satellite can be rotated arbitrarily, such that $\upsilon_u = \upsilon_{\tilde{u}} = 1$, i.e., the fluid antenna ports are positioned in-line with the propagation path of the signal, as well as the interference.

\subsection{Satellite CUMA Signal Distribution}
We find the following lemmas are useful in simplify the expression of the received in-phase CUMA signal, $\alpha _{\mathrm{I}}$. 
\begin{lemma}\label{Lemma_Signal_Expression}
     The aggregation over all the positive in-phase channels of the desired signal magnitude $\sqrt{\alpha _{\mathrm{I}}}$ can be written as 
    \begin{align}\label{Signal_Expression}
        \sqrt{\alpha_{\mathrm{ I}}} = \zeta_{u}^{\frac{1}{2}}W\sum_{ k_{\mathrm{low}}}^{k_{\mathrm{up}}}\cos{\left(\psi_{u}+\frac{2\pi}{\mu}(k-1)\right)},
    \end{align}
    where \begin{align}\label{k_Low_up_Expression}
        \begin{cases}
            k_{\mathrm{low}} &= \left\lceil \left(\frac{3}{4}-\frac{\psi_{u}}{2\pi}\right)\mu\right\rceil+1,  \\
            k_{\mathrm{up}} &=~\left\lfloor \left(\frac{5}{4}-\frac{\psi_{u}}{2\pi}\right)\mu \right\rfloor+1.
        \end{cases}
    \end{align}
    The operators $\lceil \cdot \rceil$ and $\lfloor \cdot \rfloor$ represent the ceiling and flooring operations, respectively. 
\end{lemma}
\begin{proof}
See Appendix \ref{appendix:Proof_Lemma_Signal_Expression}.
\end{proof}
\begin{lemma}\label{Lemma_Upper_Limit_Expression}
    Given $\mu$ as an even integer, the upper limit of the summation in Lemma \ref{Lemma_Signal_Expression},  $k_{\mathrm{up}}$, can be expressed as a function of its lower limit $k_{\mathrm{low}}$, given by
    \begin{align}\label{Expression_k_up}
        k_{\mathrm{up}} = k_{\mathrm{low}} + \frac{\mu}{2}-1.
    \end{align}
\end{lemma}
\begin{proof}
    See Appendix \ref{appendix:Proof_Lemma_Upper_Limit_Expression}.
\end{proof}
\begin{lemma}\label{Lemma:Number_Activated_port}
    The number of activated ports, $\bar{K}$, is given by \begin{align}\label{Activated_Ports_Expression}
        \bar{K} = \frac{K-1}{2}.
    \end{align}
\end{lemma}
\begin{proof}
    Following the same process as Appendix \ref{appendix:Proof_Lemma_Signal_Expression}, it is easy to show that $\bar{K} = W\left(k_{\mathrm{up}}-k_{\mathrm{low}}+1\right)$. Then using the result in Lemma \ref{Lemma_Upper_Limit_Expression}, we obtain \eqref{Activated_Ports_Expression}.
\end{proof}

Using Lemma \ref{Lemma_Signal_Expression} and \ref{Lemma_Upper_Limit_Expression}, the expression for the received CUMA signal power $\alpha _{\mathrm{I}}$ can be further simplified as stated in the following theorem.
\begin{theorem}\label{Theorem:Simplified_Signal_Power}
    The received CUMA in-phase signal power $\alpha_{\mathrm{I}}$ is given by \begin{align}\label{eq:Approximated_Expression_CUMA_Signal_Power}
    \alpha _{\mathrm{I}}=\frac{\zeta_{u}\cos^2{\left(-2\pi t-\frac{\pi}{\mu}+\frac{2\pi}{\mu}\left\lceil t\mu \right\rceil\right)}}{V^2},
    \end{align}
    where $t = \frac{3}{4}-\frac{\psi_{u}}{2\pi}$ follows a continues uniform distribution as $\mathcal{U}(-\frac{1}{4},\frac{3}{4})$ and $V = \frac{\sin{\left(\frac{\pi}{\mu}\right)}}{W}$. 
\end{theorem}
\begin{proof}
    See Appendix \ref{appendix:Proof_Theorem_Simplified_Signal}.
\end{proof}
\begin{remarkbox}\label{Remak_V}
    For sufficiently large $\mu$, with the small-angle approximation \cite{Small_Angle_Approximation}, the parameter $V$ in Theorem \ref{Theorem:Simplified_Signal_Power} can be approximated as $V \approx \frac{\pi}{\mu W} = \frac{\pi}{K-1}$.
\end{remarkbox}

With the received CUMA in-phase signal power expression provided in Theorem \ref{Theorem:Simplified_Signal_Power}, we establish the following theorem to characterize its distribution. 

\begin{theorem}\label{Theorem_Distribution_Signal}
The PDF of the received CUMA in-phase signal power $\alpha_{\mathrm{I}}$ is given by
\begin{align} \label{Eq:PDF_Signal_Power}
f_{\alpha_{\mathrm{I}}}(\alpha) = \frac{\mu}{2\pi}\sqrt{\frac{V^2}{\zeta_{u}\alpha-V^2\alpha^2}},
\end{align}
where $\alpha \in \left[\frac{\cos^2{\left(\frac{\pi}{\mu}\right)}\zeta_{u}}{V^2},\frac{\zeta_{u}}{V^2}\right]$ is the random variable. 
\end{theorem}
\begin{proof}
See Appendix \ref{Appendix:Proof_Theorem_Signal_Distribution}.
\end{proof}
\begin{remarkbox}\label{Remark:Constant_Signal}
   In the case of the compact fluid antenna, i.e., a sufficiently large $\mu$, according to the small angel approximation \cite{Small_Angle_Approximation}, we have $\cos^2{\left(\frac{\pi}{\mu}\right)} \approx 1$. Therefore, the received CUMA in-phase signal power approaches a constant as $\alpha_{\mathrm{I}} = \frac{\zeta_{u}}{V^2}$. This indicates that, under the condition of a highly compact fluid antenna, the received CUMA signal power is independent from the received phase at the reference port. 
\end{remarkbox}

\subsection{Satellite CUMA Interference Distribution}
In this subsection, we derive the expression for the received CUMA in-phase interference power from user $\tilde{u}$, denoted as $Y_{\tilde{u}}$, along with its conditional CDF and PDF. Subsequently, we present the PDF of the total interference power, $\beta_{\mathrm{I}}$.

\begin{theorem}\label{Theorem:Interference_Expression}
    The received CUMA in-phase interference power from user $\tilde{u}$, $Y_{\tilde{u}}$, is given by \begin{align}\label{eq:Interference_Expression}
    Y_{\tilde{u}}=\frac{\zeta_{\tilde{u}}\sin^2{\left(\psi_{\tilde{u}}-\frac{\pi}{\mu}+\frac{2\pi}{\mu}\left\lceil t\mu \right\rceil\right)}}{V^2},
    \end{align}
    where $\psi_{\tilde{u}}$ is the phase of the received interference from user $\tilde{u}$ at the reference port of the fluid antenna, $t$ and $V$ are defined in Theorem \ref{Theorem:Simplified_Signal_Power}.
\end{theorem}
\begin{proof}
    This proof follows a similar approach to that in \eqref{Appendix_B_1} from Appendix \ref{appendix:Proof_Theorem_Simplified_Signal}.  Due to space limit, the derivation is omitted here.
\end{proof}
\begin{theorem}
    Conditioned on $t$, the conditional cumulative distribution function (CDF) of $Y_{\tilde{u}}$ is given by
    \begin{align}\label{Contional_PDF_Interference}
        F_{Y_{\tilde{u}}}(Y | t) =1- \frac{\arccos{(\frac{2V^2}{\zeta_{\tilde{u}}}Y-1)}}{\pi}.
    \end{align}
    where $Y$ is the threshold and $t$ is given in Theorem \ref{Theorem:Simplified_Signal_Power}.
\end{theorem}
\begin{proof}
    This expression can be derived based from the CDF definition of random variable. Due to space limit, the derivation is omitted here.
\end{proof}
\begin{remarkbox}\label{Remark:independence}
    The received CUMA in-phase interference power from user $\tilde{u}$, $Y_{\tilde{u}}$, is independent from the in-phase signal power $\alpha_{\mathrm{I}}$, as indicated by its conditional CDF in \eqref{Contional_PDF_Interference}. 
\end{remarkbox}
\begin{corollarybox}\label{Coro_PDF_Intereference_each_user}
    The PDF of the received CUMA in-phase interference power from user $\tilde{u}$, $Y_{\tilde{u}}$, is given by
    \begin{align}\label{PDF_interference_power_each_user}
        f_{Y_{\tilde{u}}}(y) =  \frac{1}{\pi}\sqrt{\frac{V^2}{\zeta_{u}y-V^2y^2}},
    \end{align}
   where $y \in \left[0,\frac{\zeta_{u}}{V^2}\right]$ is the random variable.
\end{corollarybox}
\begin{proof}
    This result can be derived by taking the derivative of \eqref{Contional_PDF_Interference} with respect to $Y$ with further simplification.
\end{proof}

The aforementioned discussion provide the distribution of interference power when a single interference user exist, in the following, we provide the distribution of interference power when a large number of interference user exists, massive access.

\begin{theorem}\label{Theorem:Interference_CIT}
    For a sufficiently large $U$, the received overall CUMA in-phase interference power $\beta _{\mathrm{I}}$ is approximated by a truncated Gaussian distribution, with a PDF given by
   \begin{align}\label{Eq: Interference_PDF}
    f_{\beta_{\mathrm{I}}}(\beta) = \frac{1}{\Phi{\left(\frac{\omega}{\kappa}\right)}}\frac{1}{\sqrt{2\pi\kappa^2}}e^{-\frac{\left(\beta-\omega\right)^2}{2\kappa^2}},
    \end{align}
    where
    \begin{equation}
    \left\{\begin{aligned}
    \omega &= \frac{1}{2 V^2} \sum_{\substack{\tilde{u}=1 \\ \tilde{u} \ne u}}^{U} \zeta_{\tilde{u}},\\
    \kappa &= \sqrt{\frac{1}{8 V^4}\sum_{\substack{\tilde {u}=1\\ \tilde {u}\ne u}}^{U}\zeta_{\tilde{u}}^2},
    \end{aligned} \right.
    \end{equation}
    represents the mean and standard derivation of the untruncated Gaussian, respectively. $\Phi(\cdot)$ is the CDF of standard Gaussian distribution, given by
    \begin{align}
        \Phi(x) = \frac{1}{2\pi}\int_{-\infty}^{x}e^{-\frac{y^{2}}{2}}dy.
    \end{align}
\end{theorem}
\begin{proof}
    See Appendix \ref{appendix:Proof_Interference_CIT}. 
\end{proof}

With Theorem \ref{Theorem:Interference_CIT} providing the distribution of the overall interference, we have the following corollary that captures the distribution of overall interference power plus noise.
\begin{corollarybox}\label{Coro:PDF_Interference_Plus_Noise}
    The PDF of overall interference power plus noise, $\tilde{\beta_{\mathrm{I}}} = \beta_{\mathrm{I}}+\frac{\bar{K}}{2\Gamma}$, is given by
    \begin{align}\label{Eq:PDF_Interference_Plus_Noise}
        f_{\tilde{\beta_{\mathrm{I}}} }(\tilde{\beta}) = \frac{1}{\Phi{\left(\frac{\omega}{\kappa}\right)}}\frac{1}{\sqrt{2\pi\kappa^2}}e^{-\frac{\left(\tilde{\beta}-\omega-\frac{\bar{K}}{2\Gamma}\right)^2}{2\kappa^2}},
    \end{align}
    where $\tilde{\beta}$ is the random variable that follows a truncated Gaussian distribution.
\end{corollarybox}
\begin{proof} 
    This result can be derived from Theorem \ref{Theorem:Interference_CIT} and using the PDF location-scale transformation provided in \cite[Example 8.1.4]{Probability_Book}.
\end{proof}

\subsection{Satellite CUMA SINR Analysis}
With the distributions of signal and interference power established above, we now turn to investigate the distribution of the CUMA SINR. To begin with, we have the following theorem that characterizes its PDF.
\begin{theorem}\label{Theorem:PDF_Z}
    The PDF of the SINR variable $Z_{\mathrm{I}} = \frac{\alpha_{\mathrm{I}}}{\tilde{\beta_{\mathrm{I}}}}$ is given by
    \begin{multline}\label{Eq:PDF_Z}
    f_{Z_{\mathrm{I}}}(z) = \frac{\mu}{2\pi}\frac{1}{\Phi{\left(\frac{\omega}{\kappa}\right)}}\frac{1}{\sqrt{2\pi\kappa^2}}  \\ \times\int_{\frac{\cos^2{\left(\frac{\pi}{\mu}\right)}\zeta_{u}}{z V^2}}^{\frac{\zeta_{u}}{z V^2}}\sqrt{\frac{\tilde{\beta}V^2}{z\zeta_{u}-z^2\tilde{\beta}V^2}} 
    e^{-\frac{\left(\tilde{\beta}-\omega-\frac{\bar{K}}{2\Gamma}\right)^2}{2\kappa^2}}d\tilde{\beta}.
    \end{multline}
\end{theorem}
\begin{proof}
    See Appendix \ref{Appendix:Proof_PDF_Z}.
\end{proof}

Theorem \ref{Theorem:PDF_Z} provide the exact PDF of of the SINR. However, it is given in an integral-form. Considering a very compact fluid antenna, we provide the simplified result as follows.

\begin{theorem}\label{Theorem:PDF_SINR_Large_mu}
    Given a sufficiently large $\mu$, the PDF of the SINR variable $Z_\mathrm{I} = \frac{\alpha_{\mathrm{I}}}{\tilde{\beta_{\mathrm{I}}}}$ can be approximated as
    \begin{align}\label{Eq:PDF_Z_larger_mu}
        f_{Z_{\mathrm{I}}}(z) = \frac{\zeta_{u}}{\Phi{\left(\frac{\omega}{\kappa}\right)}V^2}\frac{1}{z^2\sqrt{2\pi\kappa^2}}e^{-\frac{\left(\frac{\zeta_{u}}{zV^2}-\omega-\frac{\bar{K}}{2\Gamma}\right)^2}{2\kappa^2}}.
\end{align}
\end{theorem}
\begin{proof}
    According to Remark \ref{Remark:Constant_Signal}, for a sufficiently large $\mu$, the CUMA in-phase signal power approaches a constant, given by $\alpha_{\mathrm{I}} = \frac{\zeta_{u}}{V^2}$. Subsequently, the PDF of $Z_{\mathrm{I}} = \frac{\alpha_{\mathrm{I}}}{\tilde{\beta_{\mathrm{I}}}} = \frac{\zeta_{u}}{V^2\tilde{\beta_{\mathrm{I}}}}$ can be derived as 
    $f_{Z_{\mathrm{I}}}(z) = f_{\tilde{\beta_{\mathrm{I}}}}\left(\frac{\zeta_{u}}{zV^2}\right)\left|\frac{d{\tilde{\beta_{\mathrm{I}}}}}{dz_{\mathrm{I}}}\right|$, where $f_{\tilde{\beta_{\mathrm{I}}}}(\cdot)$ is provided in Corollary \ref{Coro:PDF_Interference_Plus_Noise}.
\end{proof}

\subsection{Performance Metrics}
Using the provided PDF expressions of the SINR, we now proceed to derive some performance metrics. We first have the following remarks.
\begin{remarkbox}
    The mean of SINR variable $Z_{\mathrm{I}}$ is given by
    \begin{align}\label{eq:Mean_SINR}
    \mathrm{E}\left[Z_{\mathrm{I}}\right] = \int_{\frac{\cos^2{(\frac{\pi}{\mu})\zeta_{u}}}{(U-1)\cos^2{(\frac{\pi}{\mu})\zeta_{u}}+\Gamma V^2}}^{\frac{\zeta_{u}}{\Gamma V^2}}z f_{Z_{\mathrm{I}}}(z)dz,
    \end{align}
    where $f_{Z_{\mathrm{I}}}(z)$ can be chosen from either \eqref{Eq:PDF_Z} or \eqref{Eq:PDF_Z_larger_mu}.
\end{remarkbox}
\begin{remarkbox}
    The mean of SNR variable $S_{\mathrm{I}} = \frac{2\Gamma}{\bar{K}}\alpha_{\mathrm{I}}$ is given by
    \begin{align}
    \mathrm{E}\left[S_{\mathrm{I}}\right] = \frac{4\Gamma}{K-1}\int_{\frac{\cos^2{\left(\frac{\pi}{\mu}\right)}\zeta_{u}}{V^2}}^{\frac{\zeta_{u}}{V^2}}\alpha f_{\alpha_{\mathrm{I}}}(\alpha)d\alpha,
    \end{align}
    where $f_{\alpha_{\mathrm{I}}}(\alpha)$ is given in \eqref{Eq:PDF_Signal_Power}.
\end{remarkbox}

It is worth noting that the exact PDF of $Z_{\mathrm{I}}$ is provided in an integral-form. Consequently, deriving the exact outage probability of the CUMA network directly from $Z_{\mathrm{I}}$ would result in a double integral expression. To simplify the expression for the exact outage probability, the following corollary evaluates it in a single integral-form.

\begin{corollarybox}\label{Coro:OP_Single_Integral}
    In satellite CUMA network, the exact outage probability of the typical user is given by \eqref{Eq:OP_Single_Integral}, shown at the top of the next page, where $\gamma$ is the threshold and $\Phi(\cdot)$ is given in Theorem \ref{Theorem:Interference_CIT}.
\end{corollarybox}
    \begin{figure*}
     \begin{align}\label{Eq:OP_Single_Integral}
        {\mathcal O}_{\mathrm{I}}\left(\gamma\right) = \frac{1}{\Phi{\left(\frac{\omega}{\kappa}\right)}}-\frac{\mu}{2\pi\Phi{\left(\frac{\omega}{\kappa}\right)}}\int_{\alpha=\frac{\cos^2{\left(\frac{\pi}{\mu}\right)}\zeta_{u}}{V^2}}^{\frac{\zeta_{u}}{V^2}}\sqrt{\frac{V^2}{\zeta_{u}\alpha-V^2\alpha^2}}\Phi{\left(\frac{\frac{\alpha}{\gamma}-\omega-\frac{\bar{K}}{2\Gamma}}{\kappa}\right)}d\alpha.
    \end{align}
        \hrulefill
    \end{figure*}
\begin{proof}
    See Appendix \ref{Appendix:Proof_OP_Single_Integral}.
\end{proof}

The exact outage probability is now expressed in a single integral-form. For expression simplicity and to provide further insights, we have the following corollary, which evaluates the approximated outage probability in a closed form.
\begin{corollarybox}\label{Coro:OP_Large_mu}
    In satellite CUMA network, given a sufficiently large $\mu$, the outage probability of the typical user can be approximated as
    \begin{align}\label{Eq:OP_Large_mu}
        {\mathcal O}_{\mathrm{I}}\left(\gamma\right) = \frac{1-\Phi\left(\frac{\zeta_{u}}{\gamma\kappa V^2}-\frac{\bar{K}}{2\kappa\Gamma}-\frac{\omega}{\kappa}\right)}{\Phi\left(\frac{\omega}{\kappa}\right)}.
    \end{align}
    where $\Phi(\cdot)$ is given in Theorem \ref{Theorem:Interference_CIT}.
\end{corollarybox}
\begin{proof}
    This result can be derived based on Theorem \ref{Theorem:PDF_SINR_Large_mu}, due to space limit is omitted here.
\end{proof}

\begin{corollarybox}
    The ergodic rate of the satellite CUMA network can be expressed as \begin{multline}\label{Eq:Ergodic_Rate_Exactly}
         C_{\mathrm{e}} = \frac{UB}{\ln{(2)}}\int_{0}^{\infty}\frac{1}{1+y}\big(1- {\mathcal O}_{\mathrm{I}}\left(y\right)\big)dy~~({\rm bits/sec}),
    \end{multline}
    where ${\mathcal O}_{\mathrm{I}}\left(\cdot\right)$ is the outage probability that can be chosen from the expression \eqref{Eq:OP_Single_Integral} or \eqref{Eq:OP_Large_mu}.
\end{corollarybox}
\begin{proof}
    According to \cite[Lemma 1]{Performace_Limits_FAS}, the ergodic capacity can be obtained as
\begin{align}\label{Eq:Lemma_1_Cite}
    C_{\mathrm{e}} = \frac{UB}{\ln{(2)}}\int_{0}^{\infty}\frac{1}{1+y}\mathbb{P} \left\{ z> y\right\}dy.
\end{align}
Subsequently, substituting $\mathbb{P} \left\{ z> y\right\} = 1- \mathbb{P} \left\{ z\leq y\right\}$ and employing the result from either Corollary \ref{Coro:OP_Single_Integral} or Corollary \ref{Coro:OP_Large_mu} into \eqref{Eq:Lemma_1_Cite}, we obtain \eqref{Eq:Ergodic_Rate_Exactly}.
\end{proof}

\section{Physical Insights}\label{Sec:Physical_Insights}
In this section, we aim to gain insights into the satellite CUMA network, including the beamforming gain it provides and its performance compared to conventional MRC. Subsequently, we extened the presented results to discuss other port activation strategies. Different from the previous CUMA network in \cite{CUMA}, which suggests using two RF chains to exploited CUMA's potential, we demonstrate that, with a sufficiently large $K$, a single RF chain is sufficient to fully realize the potential of the satellite CUMA network.

\subsection{Beamforming Gain}\label{Discussion_beamfroming_gain}
Recall the SINR expression in \eqref{SINR_In_Phase_K_1}, with the results provided in Theorem \ref{Theorem:Simplified_Signal_Power} and \ref{Theorem:Interference_Expression}, the SINR expression can be expressed as
\begin{align}\label{Eq:SINR_exact}
    &{\mathrm{SINR}}_{\mathrm{I}}^{\mathcal{K}_1} = \\ &\frac{\frac{4}{\mu}\frac{W}{\sin^2{\left(\frac{\pi}{\mu}\right)}}\zeta_{u}\cos^2{\left(-2\pi t-\frac{\pi}{\mu}+\frac{2\pi}{\mu}\left\lceil t\mu \right\rceil\right)}}{\frac{4}{\mu}\frac{W}{\sin^2{\left(\frac{\pi}{\mu}\right)}}\sum_{\substack{\tilde {u}=1\\ \tilde {u}\ne u}}^{U}\zeta_{\tilde{u}}\sin^2{\left(\psi_{\tilde{u}}-\frac{\pi}{\mu}+\frac{2\pi}{\mu}\left\lceil t\mu \right\rceil\right)}+\frac{1}{\Gamma}}.\nonumber
\end{align}
However, the beamforming gain for signal is hard to be characterizes according to \eqref{Eq:SINR_exact}. Consequently, consider a compact fluid antenna configuration, along with the results presented in Lemma \ref{Lemma:Number_Activated_port}, Remark \ref{Remak_V} and \ref{Remark:Constant_Signal}, the SINR expression for the satellite CUMA network can be approximated as
\begin{align}\label{Eq:SINR_Lagr_mu}
    {\mathrm{SINR}}_{\mathrm{I}}^{\mathcal{K}_1} \approx \frac{4\zeta_{u}\frac{K-1}{\pi^2}}{4\frac{K-1}{\pi^2}\sum_{\substack{\tilde {u}=1\\ \tilde {u}\ne u}}^{U}\zeta_{\tilde{u}}\sin^2{\left(\psi_{\tilde{u}}-\frac{\pi}{\mu}+\frac{2\pi}{\mu}\left\lceil t\mu \right\rceil\right)}+\frac{1}{\Gamma}}.
\end{align}

As a performance benchmark, under the LoS scenario, the SINR expression for MRC equipped with $M$ antennas (RF chains) can be derived through \cite[(4.31)]{Emil_Book}, expressed as
\begin{align}\label{Eq:MRC_SINR}
{\mathrm{SINR}}_{\mathrm{MRC}} &= \frac{M\zeta_{u}}{M\sum_{\substack{\tilde {u}=1\\ \tilde {u}\ne u}}^{U}\zeta_{\tilde{u}}\left|e^{-j\psi_{u}+j\psi_{\tilde{u}}}\right|^2+\frac{1}{\Gamma}} \nonumber \\
& = \frac{M\zeta_{u}}{M\sum_{\substack{\tilde {u}=1\\ \tilde {u}\ne u}}^{U}\zeta_{\tilde{u}}+\frac{1}{\Gamma}}.
\end{align}

From \eqref{Eq:SINR_Lagr_mu}, it can be observed that, in the case of CUMA, a beamforming gain (a.k.a power gain) of $4\frac{K-1}{\pi^2}$ is achieved for the desired signal, which increases linearly with the number of ports $K$. K. According to \eqref{Eq:MRC_SINR}, the beamforming gain of desired signal achieved by MRC is $M$, the same as the number of antennas, as well as the number of RF chains, which is not only expensive but also heavy for satellite communications.

Moreover, under the considered scenario, MRC provides a beamforming gain of $M$ to the interference, identical to that for signal. It indicates that MRC cannot distinguish the difference of signal and interferences if the angle of arrival exhibits similarity. In other words, as MRC is also known as the spatial filter,  it can only suppress interference that can be spatially distinguished from the signal, as discussed in \cite{Emil_Book}. This indicates that, for a satellite communications scenario, especially for high-orbit satellite communications, where the user may be seen as spatially compact, the use of MRC on uplink communication only suppresses noise power when using multiple antennas. In contrast, the beamforming gain for interfering user $\tilde{u}$, provided by CUMA, is  $4\frac{K-1}{\pi^2}\sin^2{\left(\psi_{\tilde{u}}-\frac{\pi}{\mu}+\frac{2\pi}{\mu}\left\lceil t\mu \right\rceil\right)}$, for which the precise value is determined by the phase difference of signal and interference observed at the reference port. This performance difference is due to that satellite CUMA builds upon the absolute phase difference between each ports, while MRC builds upon the relative phase difference between each antennas.

Given the preceding discussion, we now present the following corollary, which compares the beamforming gain provided by CUMA and MRC.
\begin{corollarybox}\label{Coro:CUMA_Outperform_MRC}
    In the considered system, CUMA outperforms MRC when
    \begin{align}\label{Eq:CUMA_Outperform_MRC}
     K > \max\left\{\left\lceil\frac{\frac{\pi^2}{4}M}{M\Gamma\sum_{\substack{\tilde {u}=1\\ \tilde {u}\ne u}}^{U}\zeta_{\tilde{u}}\delta_{\tilde{u}}+1}\right\rceil,
     \Big\lceil \epsilon W \Big\rceil\right\}+1,
    \end{align}
    where $\delta_{\tilde{u}} \triangleq 1-\sin^2{\left(\psi_{\tilde{u}}-\frac{\pi}{\mu}+\frac{2\pi}{\mu}\left\lceil t\mu \right\rceil\right)}$, $\epsilon$ is the pre-set threshold for the fluid antenna port density $\mu$.
\end{corollarybox}
\begin{proof}
    By utilizing \eqref{Eq:SINR_Lagr_mu} and \eqref{Eq:MRC_SINR} and under the condition that $\mu \geq \epsilon$, this result is obtained.
\end{proof}

Note the precise value for $K$ that CUMA outperforms MRC is related to the interference term, as reflected in $\delta_{\tilde{u}}$ and $\zeta_{\tilde{u}}$. To provide more insights into the performance difference between CUMA and MRC, we set $\zeta_{u} = \zeta_{\tilde{u}}$ in the following discussion. The following remarks present the exact value of $K$ at which CUMA outperforms MRC under different scenarios.

\begin{remarkbox}\label{Remark:CUMA_Outperform_MRC_Noise}
    In the worst-case for CUMA, the phase of the interference aligns with that of the signal, i.e., $\psi_{u} = \psi_{\tilde{u}}$, resulting in $\delta_{\tilde{u}} = 1$. CUMA outperforms MRC under the condition of
    \begin{align}
     K > \max\left\{\left\lceil\frac{\pi^2}{4}M\right\rceil, \Big\lceil \epsilon W \Big\rceil\right\}+1.
    \end{align}
    This is also the condition for CUMA to outperform MRC under a noise-limited scenario.
\end{remarkbox}

Remark \ref{Remark:CUMA_Outperform_MRC_Noise} also corresponds to the special case of a noise-limited scenario, which is particularly relevant in satellite communications due to the extremely large propagation distance. The following remark addresses the scenario dominated by a LoS propagation path and an interference-limited environment, which is commonly encountered in short-range communication systems operating in the millimeter-wave and higher frequency bands.

\begin{remarkbox}\label{Remark:CUMA_Outperform_MRC_Interference}
In a LoS and interference-limited scenario, CUMA consistently outperforms MRC under the compact fluid antenna condition, i.e.,
\begin{align}\label{Eq:CUMA_Outperform_MRC}
K > \Big\lceil \epsilon W \Big\rceil + 1,
\end{align}
which results from the fact that $\delta_{\tilde{u}} \leq 1$.
\end{remarkbox}

It is worth noting that the conditions outlined in Corollary \ref{Coro:CUMA_Outperform_MRC}, Remark \ref{Remark:CUMA_Outperform_MRC_Noise} and Remark \ref{Remark:CUMA_Outperform_MRC_Interference} for CUMA to outperform MRC are based on the compact fluid antenna configuration. In practice, however, CUMA may still outperform MRC without the compact assumption.

\subsection{Compact or non-Compact Fluid Antenna?}
 Despite the aforementioned discussion provides the precise condition under which CUMA outperforms MRC, in terms of beamforming gain, it is built upon the compact fluid antenna assumption, i.e., sufficient large $\mu$. Thus, it can not provide the insights into whether compact fluid antenna is better, or non-compact fluid antenna is better in the considered system. To answer this question, here, we focus on the stochastic point of view, and evaluate the performance gap between the received CUMA signal and interference power from each user. To begin with, we have the following definition check the first-order stochastic dominance.
\begin{definition}[First-Order Stochastic Dominance \textnormal{\cite[Definition 2]{FSD}}]\label{FSD_Definition}
Let $\mathcal{G}$ denote a set of probability distributions over income, each with a CDF. For all $g', g'' \in \mathcal{G}$, we say that $g'$ \emph{first-order stochastically dominates} $g''$, denoted $g' \succeq_{\mathrm{FSD}} g''$, if and only if their corresponding cumulative distribution functions $G'(y)$ and $G''(y)$ satisfy
\begin{align}
    G'(y) \leq G''(y) \quad \text{for all } y \in \mathbb{R},
\end{align}
with strict inequality for some $y \in \mathbb{R}$.
\end{definition}

It is evident that the first-order stochastic dominance is related to the CDF difference between random variables. We found the following lemma useful in deriving the CDF difference between the received signal power $\alpha_{\mathrm{I}}$ and the received interference power from user $\tilde{u}$, $Y_{\tilde{u}}$.
\begin{lemma}\label{lemma_rewrite_PDF}
    The PDF of the received CUMA signal power, $\alpha_{\mathrm{I}}$, can be rewritten in term of the PDF of interference power from user $\tilde{u}$, as
\begin{equation}
\begin{aligned}
f_{\alpha_{\mathrm{I}}}(\alpha) =
& \begin{cases}
\displaystyle
0,&  0 < Y \leq \frac{\cos^2\left(\frac{\pi}{\mu}\right)\zeta_u}{V^2}, \\[10pt]
\displaystyle
\frac{\mu}{2} f_{Y_{\tilde{u}}}(y)dy,&  \frac{\cos^2\left(\frac{\pi}{\mu}\right)\zeta_u}{V^2} < Y < \frac{\zeta_u}{V^2}.
\end{cases}
\end{aligned}
\end{equation}
\end{lemma}
\begin{proof}
    This is the result of Theorem \ref{Theorem_Distribution_Signal} and Corollary \ref{Coro_PDF_Intereference_each_user}.
\end{proof}

Consequently, the following theorem specific the difference between the CDF of CUMA signal power, $\alpha_{\mathrm{I}}$ and the interference power from user $\tilde{u}$, $Y_{\tilde{u}}$.
\begin{theorem}\label{Theorem_CDF_Dif}
Let $F_{\alpha_{\mathrm{I}}}(\cdot)$ and $F_{Y_{\tilde{u}}}(\cdot)$ denoted the CDF of the received CUMA signal power, and interference power from user $\tilde{u}$, respectively. The difference between $F_{\alpha_{\mathrm{I}}}(\cdot)$ and $F_{Y_{\tilde{u}}}(\cdot)$ is given by
\begin{equation}\label{eq_CDF_Dif}
\begin{aligned}
&F_{Y_{\tilde{u}}}(Y)-F_{\alpha_{\mathrm{I}}}(Y) =  \\
& \begin{cases}
\displaystyle
\int_{0}^{Y} f_{Y_{\tilde{u}}}(y)dy, & 0 < Y \leq \frac{\cos^2\left(\frac{\pi}{\mu}\right)\zeta_u}{V^2}, \\[10pt]
\displaystyle
\left(\frac{\mu}{2}-1\right) \int_{\frac{\cos^2\left(\frac{\pi}{\mu}\right)\zeta_u}{V^2}}^{\frac{\zeta_u}{V^2} - Y} f_{Y_{\tilde{u}}}(y)dy, & \frac{\cos^2\left(\frac{\pi}{\mu}\right)\zeta_u}{V^2} < Y < \frac{\zeta_u}{V^2},
\end{cases}
\end{aligned}
\end{equation}
where $Y$ is the CDF threshold and $f_{Y_{\tilde{u}}}(y)$ is the PDF expression given in \eqref{PDF_interference_power_each_user}. This difference is strictly non-negative due to the fact $\mu\geq2$.
\end{theorem}
\begin{proof}
    See Appendix \ref{Proof_Theorem_CDF_Dif}.
\end{proof}

With the CDF difference provided in Theorem \ref{Theorem_CDF_Dif}, we have the following corollary on the stochastic dominance between the received CUMA signal power, and interference power from user $\tilde{u}$.
\begin{corollarybox}\label{coro:FSD}
In the considered system, the received CUMA signal power, $\alpha_{\mathrm{I}}$, first-order stochastically dominates the interference power from user $\tilde{u}$, $Y_{\tilde{u}}$, i.e., $\alpha_{\mathrm{I}} \succ_{\text{FSD}} Y_{\tilde{u}}$.
\end{corollarybox}
\begin{proof}
    This is the result of Definition \ref{FSD_Definition} and Theorem \ref{Theorem_CDF_Dif}.
\end{proof}
\begin{corollarybox}\label{coro_PDF_Ratio}
The PDF ratio between the received CUMA signal power, and the interference power from user $\tilde{u}$, is given by
\begin{equation}
\begin{aligned}
\frac{f_{\alpha_{\mathrm{I}}}(\alpha)}{f_{Y_{\tilde{u}}}(y)} =
& \begin{cases}
\displaystyle
0,&  0 < Y \leq \frac{\cos^2\left(\frac{\pi}{\mu}\right)\zeta_u}{V^2}, \\[10pt]
\displaystyle
\frac{\mu}{2} ,&  \frac{\cos^2\left(\frac{\pi}{\mu}\right)\zeta_u}{V^2} < Y < \frac{\zeta_u}{V^2}.
\end{cases}
\end{aligned}
\end{equation}
\end{corollarybox}
\begin{proof}
    This is the result of Lemma \ref{lemma_rewrite_PDF}.
\end{proof}

The behavior of the received CUMA signal power, $\alpha_{\mathrm{I}}$, relative to the interference power from user $\tilde{u}$, $Y_{\tilde{u}}$, can be characterized in the following two respects. First, the result in Corollary \ref{coro:FSD} implies that $\alpha_{\mathrm{I}}$ is more likely to attain larger values than $Y_{\tilde{u}}$ in the sense of first-order stochastic dominance, indicating that $\alpha_{\mathrm{I}}$ statistically outperforms $Y_{\tilde{u}}$. Second, given $\mu = 2$, the performance of $\alpha_{\mathrm{I}}$ is the same as $Y_{\tilde{u}}$, as indicated in both Corollary \ref{coro:FSD} and \ref{coro_PDF_Ratio}. However, with the increase on $\mu$, $\alpha_{\mathrm{I}}$ allocates increasingly more probability mass to the upper tail of the distribution compared to $Y_{\tilde{u}}$. This density amplification strengthens the first-order stochastic dominance relationship and highlights the enhanced tendency of $\alpha_{\mathrm{I}}$ toward larger realizations. In other words, with the increase on $\mu$, the performance gap between the received CUMA signal power $\alpha_{\mathrm{I}}$ and the interference power from user $\tilde{u}$, $Y_{\tilde{u}}$, increases. This can also be understood from the perspective that the satellite CUMA network distinguishes between the desired signal and interference through absolute phase differences, as discussed in Section~\ref{Discussion_beamfroming_gain}. The denser port deployment (larger $\mu$) enables higher phase resolution, thereby enhancing signal separability and improving overall system performance\footnote{This result differs from prior studies that were based on spatially correlated FAS structures. This paper considers the LoS propagation scenario and therefore does not rely on spatial correlation.}. This suggests that the compact fluid antenna, i.e, large $\mu$, is the favorable choice in satellite communications scenario.

\subsection{Other Port Activation Strategies}
The previous discussion focused on using set $\mathcal{K}_1$ to improve the received SINR, which includes all the positive ports of the in-phase component. Initially, CUMA was proposed to activate two sets of ports, one containing the positive in-phase channels and the other containing the positive quadrature channels, thereby enabling signal diversity through the use of two distinct signals, as detailed in \cite{CUMA}. However, for the considered system, in a LoS scenario, with a sufficiently large number of ports $K$, the addition of another set of ports does not yield any improvement in SINR. To demonstrate this, we consider another set, $\mathcal{K}_2$, which contains all the negative ports of the in-phase component\footnote{The same reasoning applies when collecting all the positive (or negative) components of the quadrature signals at each port. Due to space limit, detailed discussion is omitted here. }. The rationale behind selecting $\mathcal{K}_2$ lies in ensuring that each port is assigned exclusively to either $\mathcal{K}_1$ or $\mathcal{K}_2$. This guarantees that no port is included in both sets simultaneously, i.e., $\mathcal{K}_{1} \cap \mathcal{K}_{2} = \varnothing$.
 
\begin{theorembox}\label{Theorm:Set_K_2}
    With a sufficiently large $K$, the received CUMA in-phase signal that collected through $\mathcal{K}_{2}$, is nearly identical to that collected on $\mathcal{K}_{1}$, i.e., $\alpha_{\mathrm{I},\mathcal{K}_{2}}\approx \alpha_{\mathrm{I},\mathcal{K}_{1}}$. The same argument holds for the interference, i.e., $Y_{\tilde{u},\mathcal{K}_{2}} \approx Y_{\tilde{u},\mathcal{K}_{1}}$.
\end{theorembox}
\begin{proof}
    See Appendix \ref{Appendix:Theorm_Set_K_2}.
\end{proof}

The result in Theorem \ref{Theorm:Set_K_2} indicates that ${\mathrm{SINR}}_{\mathrm{I}}^{\mathcal{K}1} \approx {\mathrm{SINR}}_{\mathrm{I}}^{\mathcal{K}_2}$ when the number of ports $K$ is large. Consequently, the combining of these two branches does not provide additional gain to the received signal, unless additional signal processing techniques can be employed with perfect knowledge of the SINR on each branch, see \cite[(40)]{CUMA}. This suggests that, for a satellite CUMA network, only a single RF chain is required for each fluid antenna.

\section{Numerical Results and Discussion}\label{sec:results}
In this section, we provide the simulation results to evaluate the performance of satellite CUMA network. The parameters used in the simulations are given in Table \ref{Table_Simulation}, unless otherwise specified. 

\begin{table}[ht]
\caption{Simulation Parameters \cite{3gpp2019study}}\label{Table_Simulation}
\centering
\begin{tabular}{p{0.4\linewidth}>{\centering\arraybackslash}p{0.4\linewidth}} 
\toprule
\textbf{Parameter} & \textbf{Value} \\ \midrule
Orbit & LEO \\
Frequency band & 30 GHz \\
User link bandwidth & $B$ = 10 MHz \\
User transmit power & $P$ = 1 W \\
Antenna Gain & $G$ = 40 dBi \\
Transmission distance & $r_{u}=r_{\tilde{u}}$ = 1200 km \\
Symbol  power & $\sigma_{s}^2$ = 1 \\
Boltzmann constant  & $c_{\mathrm{B}} = 1.381\times10^{-23}$ $\mathrm{J /K}$\\
Clear sky temperature & $T = 207^\circ$ $\mathrm{K}$ \\
\bottomrule
\end{tabular}
\end{table}

\begin{figure}[!htbp]
\centering
\includegraphics[width=0.8\linewidth]{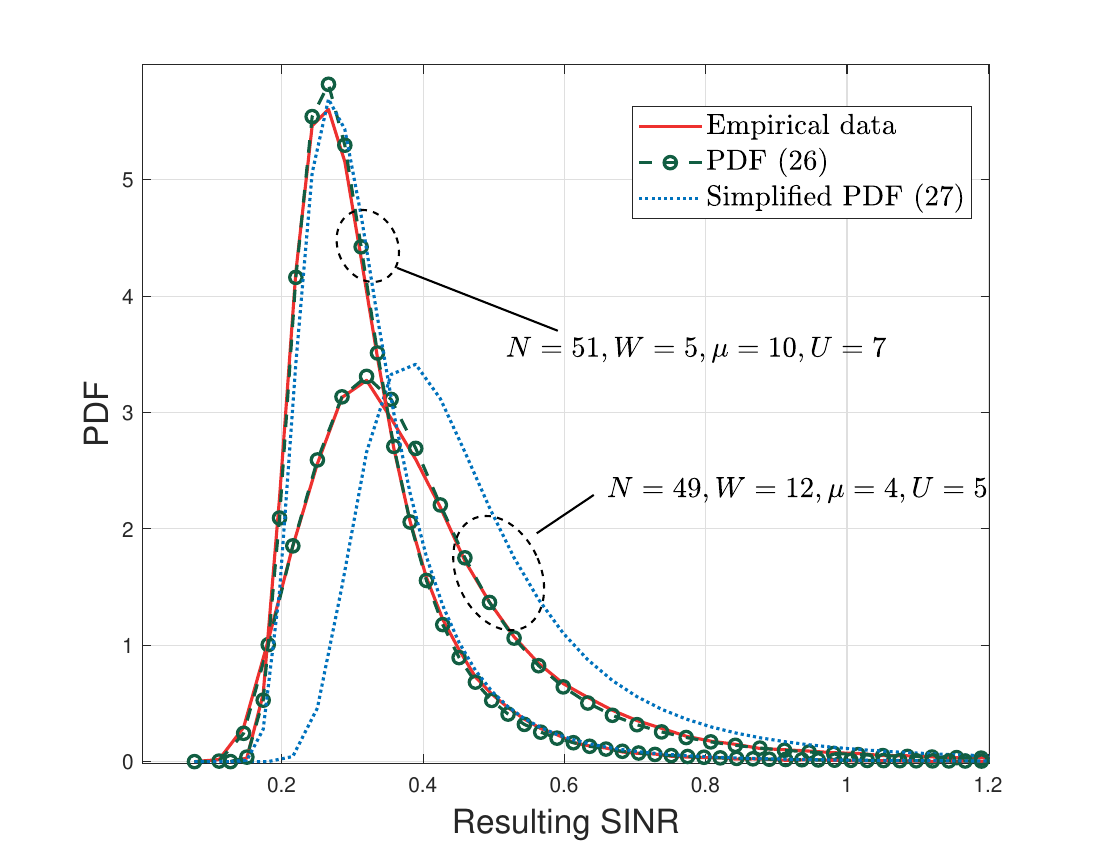}
\caption{Empirical and analytical PDF for satellite CUMA.}\label{fig:Monte_SINR}
\end{figure} 

Fig.~\ref{fig:Monte_SINR} provides the PDF of the received CUMA SINR obtained from Monte-Carlo simulations, with the analytical expression given in \eqref{Eq:PDF_Z} and the simplified expression presented in \eqref{Eq:PDF_Z_larger_mu}, evaluated under various configurations of the system parameters. Firstly, the results shows that given $\mu$ as an even integer ($\mu = 10$ and $\mu = 4$), the provided PDF expression \eqref{Eq:PDF_Z} align closely with the Monte-Carlo result. The simplified PDF only align closely with the Monte-Carlo results when $\mu = 10$ but not $\mu = 4$, demonstrated that \eqref{Eq:PDF_Z_larger_mu} is only valid for the compact fluid antenna, i.e, sufficient large $\mu$.

The results in Fig.~\ref{fig:Monte_SINR} have confirmed the accuracy of the provided PDF expression \eqref{Eq:PDF_Z_larger_mu}, under the condition that $\mu$ is an even integer. Now we turn to the results in Fig.~\ref{fig:OP_mu} that evaluate the performance of CUMA, given $\mu$ as an arbitrary integer, in terms of outage probability and average SINR. The average SINR of both MRC and zero-forcing (ZF) are provided as benchmark for comparison.

It can be observed both expression \eqref{Eq:OP_Single_Integral} and \eqref{eq:Mean_SINR} align with the Monte-Carlo results, at those point $\mu$ is an even integer. However, for odd value of $\mu$, each of these expressions tends to underestimate the performance of CUMA. Nevertheless, as the density $\mu$ increases, the discrepancy diminishes and eventually vanishes. This observation suggests that the compact expressions for CUMA signal and interference power, given by \eqref{eq:Approximated_Expression_CUMA_Signal_Power} and \eqref{eq:Interference_Expression}, as well as the PDF expression in \eqref{Eq:PDF_Z}, remain valid at the high density level, regardless of whether $\mu$ is even or not. These expressions can thus serve as a foundation for performance analysis, irrespective of whether $\mu$ is even or odd. Moreover, it is evident that as $\mu$ increases, the average SINR improves, leading to a corresponding reduction in the outage probability for the typical user. Notably, compared to CUMA and MRC, ZF results in the lowest average SINR. This is the result of low SNR conditions typical of satellite communications scenarios, along with the noise enhancement drawback inherent in ZF. While MRC achieves comparable performance to CUMA, it requires 3 antennas per user—resulting in a total of 15 antennas. In contrast, CUMA employs only 5 fluid antennas, equal to the number of users.

\begin{figure}
\centering
\includegraphics[width=0.8\linewidth]{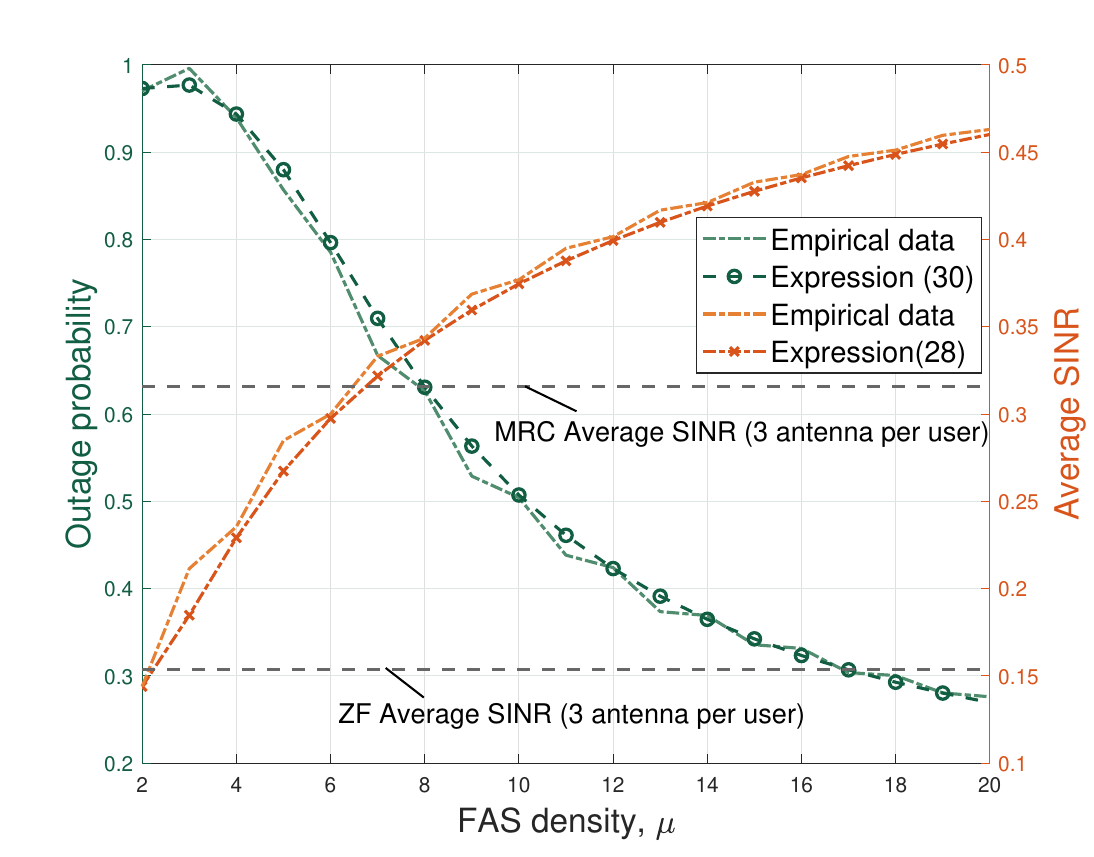}
\caption{Outage probability and average SINR against the fluid antenna port density $\mu$, given $W=2, U=5$ and $\gamma = 0.35$.} \label{fig:OP_mu}
\end{figure}

\begin{figure}
\centering
\includegraphics[width=0.8\linewidth]{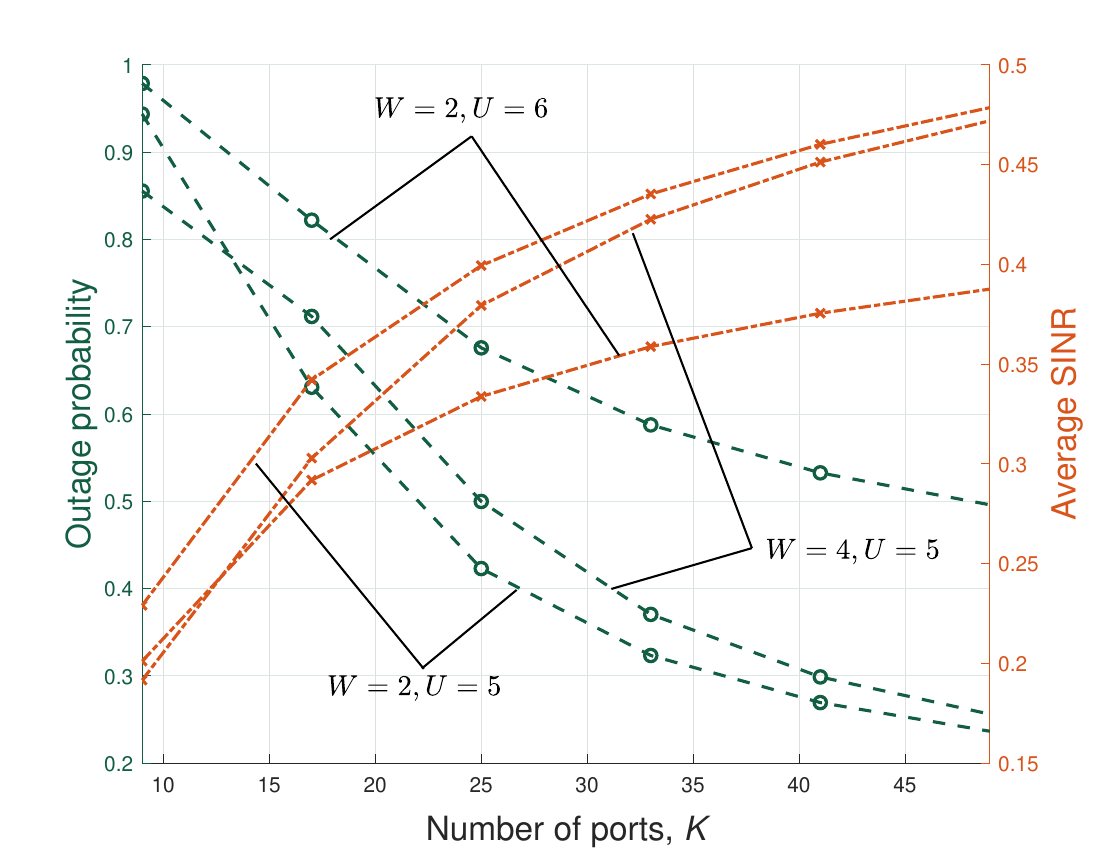}
\caption{Outage probability and average SINR against the number of ports $K$, given different $W$ and $U$ when $\gamma = 0.35$. }\label{fig:OP_N}
\end{figure}

Fig.~\ref{fig:OP_N} illustrates that as the number of ports, $K$, increases, the average SINR improves and the outage probability decreases, as increasing $K$ is equivalent to increasing the fluid antenna port density, $\mu$. Furthermore, for a fixed number of ports $K$, increasing $W$ from 2 to 4 leads to a decrease in SINR, resulting in degraded performance. This is due to given fixed $K$, increasing $W$ is equivalent to decreases $\mu$. 
In other words, it is not the larger FAS size that leads to the performance degradation, but rather the lower port density $\mu$. Additionally, as the number of users $U$, increases, the performance of the typical user decreased.

To further analyze the influence of $U$, we refer to Fig.~\ref{fig:OP_U}. It is evident that as $U$ increases, the average SINR of the typical user decreases, resulting in a higher outage probability. For the same value of $U$, the case of $\mu = 10$ outperforms $\mu = 5$, indicating that a higher fluid antenna port density helps mitigate interference. This observation is consistent with the results shown in Fig.~\ref{fig:OP_mu}.

The results in Fig.~\ref{fig:CDF_Threshold} provide the CDF of the received CUMA signal power $\alpha_{\mathrm{I}}$, and the interference power from user $\tilde{u}$, $Y_{\tilde{u}}$, under different port density $\mu$. It is evident that the CDF of $\alpha_{\mathrm{I}}$ lies entirely below to that of $Y_{\tilde{u}}$, indicating that the received CUMA signal power, $\alpha_{\mathrm{I}}$, first-order stochastically dominates the interference power $Y_{\tilde{u}}$, as discussed in Corollary \ref{coro:FSD}. Furthermore, when $\mu=6$, the increased separation between CDF curve of $\alpha_{\mathrm{I}}$ and $Y_{\tilde{u}}$, compared to the case of $\mu=4$, suggests an enhancement in the first-order stochastic dominance of
$\alpha_{\mathrm{I}}$ over $Y_{\tilde{u}}$. This highlights the effectiveness of higher FAS port density $\mu$ in mitigating interference.

\begin{figure}
\centering
\includegraphics[width=0.8\linewidth]{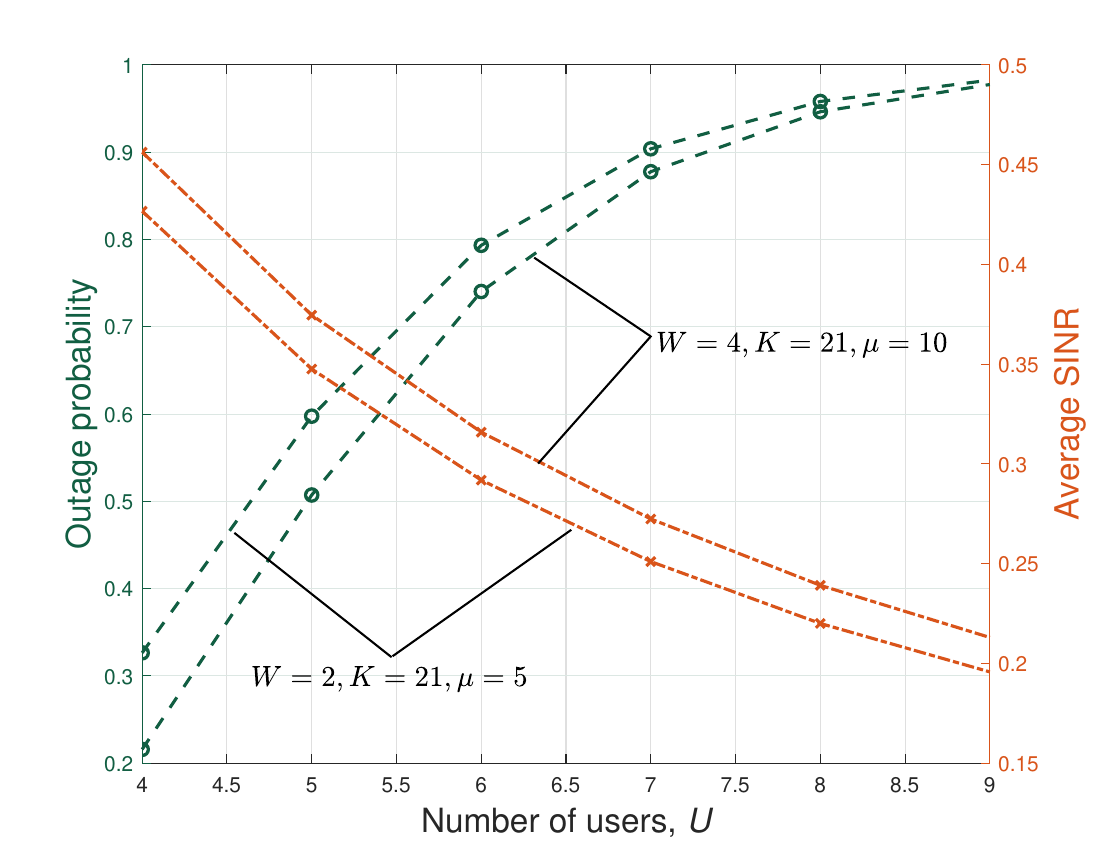}
\caption{Outage probability and average SINR against the number of users $U$, given different $\mu$ when $\gamma = 0.35$. }\label{fig:OP_U}
\end{figure}

\begin{figure}
\centering
\includegraphics[width=0.8\linewidth]{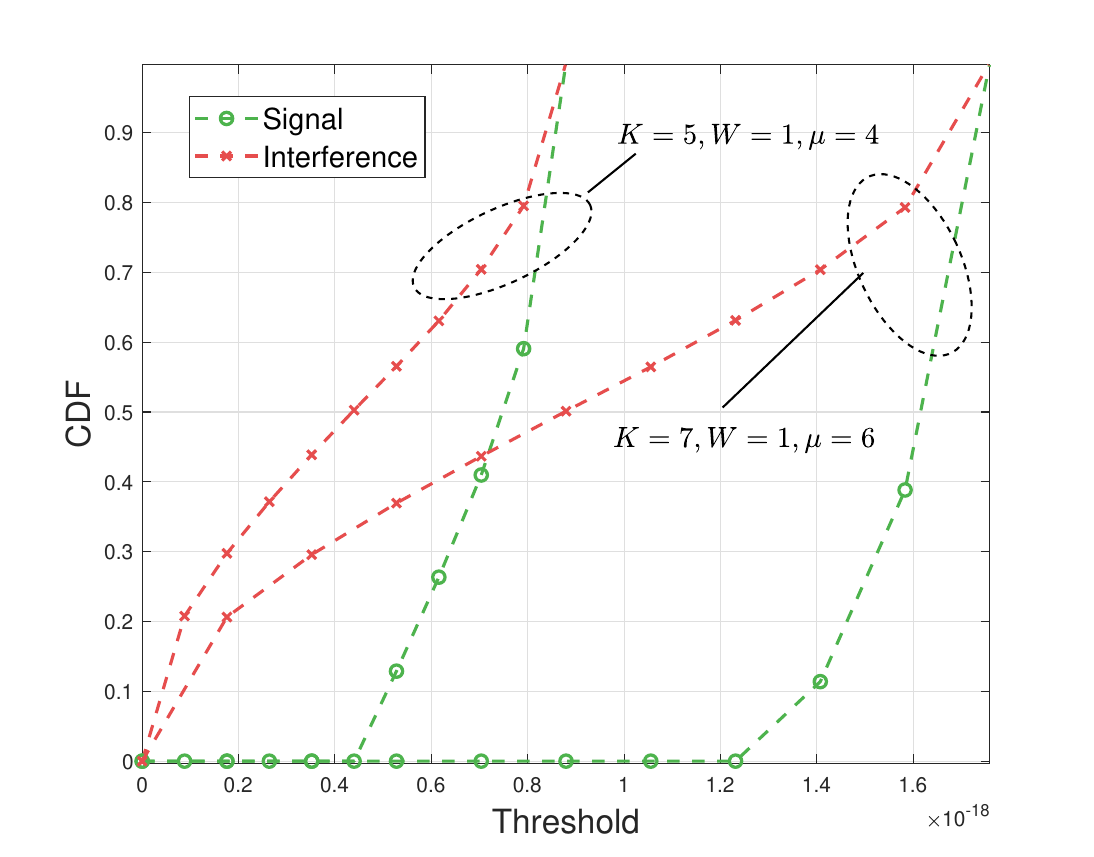}
\caption{CDF of the signal power $\alpha_{\mathrm{I}}$ and the interference power from user $\tilde{u}$, $Y_{\tilde{u}}$, against the threshold, given different $\mu$.}\label{fig:CDF_Threshold}
\end{figure}

Next, we study the performance on the ergodic rate of satellite $N$-CUMA network. As a benchmark, orthogonal CUMA ($O$-CUMA) is considered, which utilizes a single fluid antenna to service $U$ users through $U$ time slots, while the total bandwidth $B$ is allocated to each user, therefore no inter-user interference exists. As a result, $O$-CUMA achieves the same ergodic rate as a single-user CUMA system. 

\begin{figure}
\centering
\includegraphics[width=0.8\linewidth]{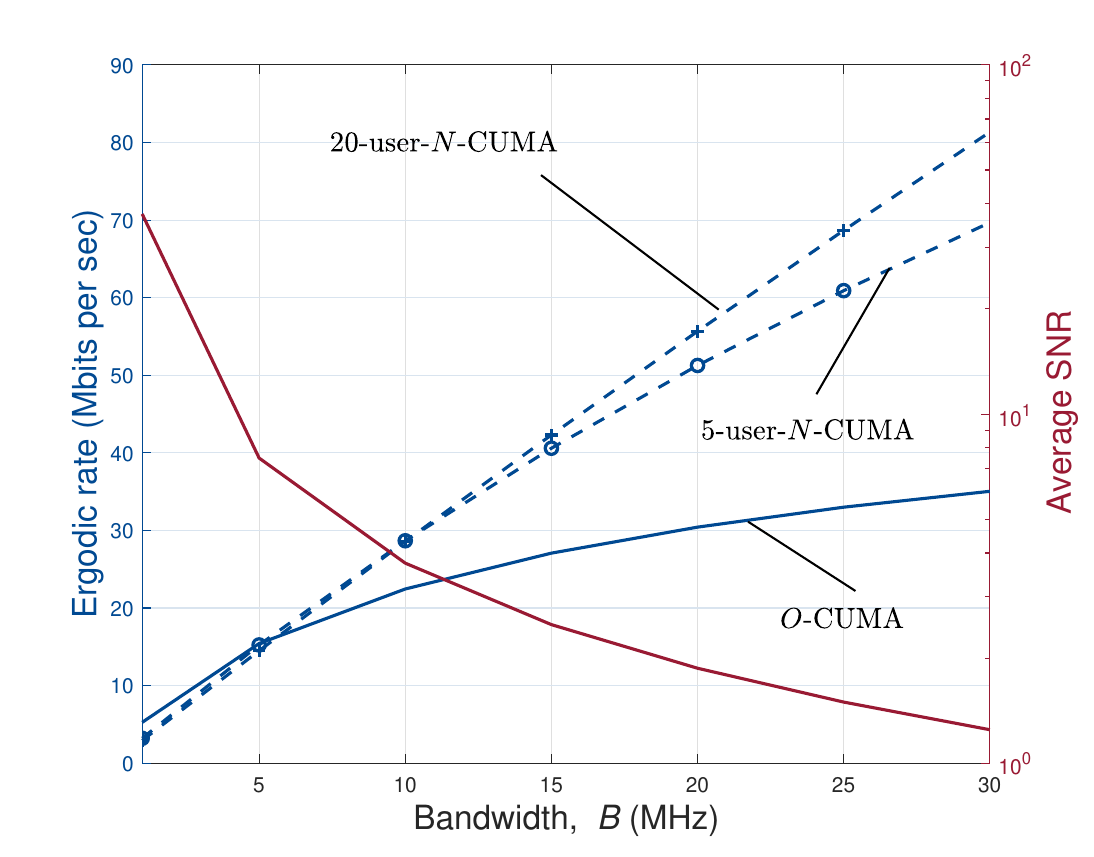}
\caption{Ergodic rate and average SNR against the bandwidth $B$, when $K=61$ and $W=3$.}\label{fig:Ce_B}
\end{figure}

The results in Fig.~\ref{fig:Ce_B} examine the impact of the total bandwidth $B$ on both ergodic rate and average SNR, given different number of users $U$. With the increase on $B$, the average SNR decreased. In contrast, the ergodic rate of three considered setting ($U=1$, $U=5$ and $U=20$) all increased. This is due to that despite the noise power is linearly increased with $B$, the larger communication bandwidth still improve the ergodic rate of CUMA system. Additionally, $O$-CUMA appears to be the optimal setting given limited $B$, but with the increase on $B$, 20-user-$N$-CUMA starts to outperform both 5-user-$N$-CUMA. It is the result that under broad bandwidth, the system tends to be noise-limited, i.e., the increase on interfering user do not decrease the SINR that much.  

\begin{figure} 
    \centering
  \subfloat[$B = 10$ MHz.\label{fig:Ce_mu_B_10}]{\includegraphics[width=0.48\linewidth,height=0.22\textheight]{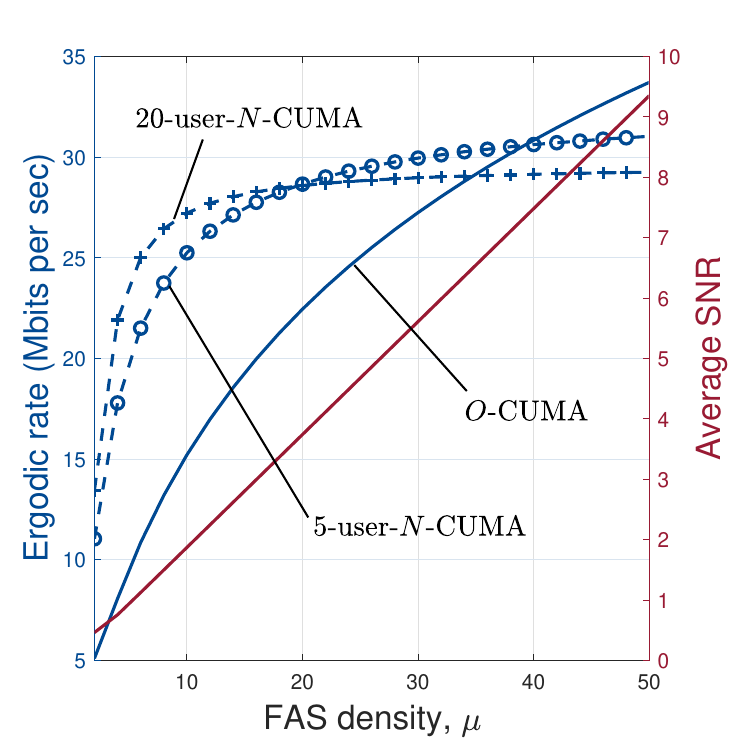}}
    \hfill
  \subfloat[$B = 20$ MHz.\label{fig:Ce_mu_B_20}]{\includegraphics[width=0.48\linewidth,height=0.22\textheight]{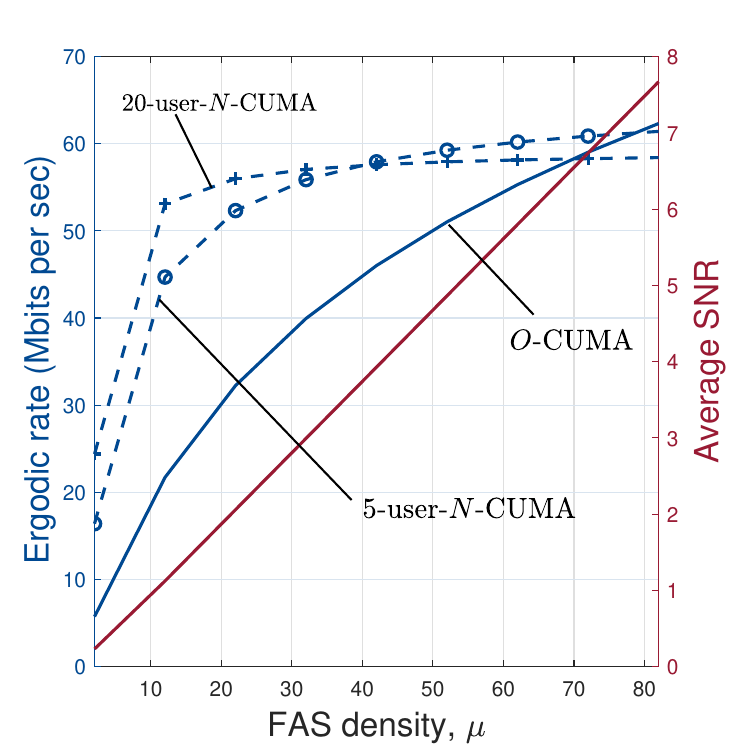}}
  \caption{Ergodic rate and average SNR against the fluid antenna port density $\mu$, given different bandwidth $B$ when $W=3$.}\label{fig:Ce_mu} 
\end{figure}

Now, the results in Fig.~\ref{fig:Ce_mu} investigate the ergodic rate and average SNR when the number of fluid antenna port density $\mu$ changes, under different bandwidth $B$ and $U$. Firstly, in both $B$ = 10 MHz and $B$ = 20 MHz, with the increase on $\mu$, the average SNR is linearly increased. This is due to that $\mu$ is another measure on the number of port $K$, while the beamforming gain is linearly increased with $K$, as discussed in Section \ref{Discussion_beamfroming_gain}. Also, at low density $\mu$, 20-user-$N$-CUMA achieves the highest ergodic rate while $O$-CUMA achieves the lowest, whereas at high $\mu$, the trend reverses. This occurs because, with high $\mu$, the system tends to become interference-limited due to the high beamforming gain applied by CUMA on both the signal and interference. This also explains why the 20-user-$N$-CUMA initially outperforms, but is eventually surpassed by the 5-user-$N$-CUMA. Moreover, it can be observed that with larger bandwidth $B$, $O$-CUMA requires higher density $\mu$ to outperform $N$-CUMA, as the higher bandwidth increases the noise power, necessitating higher $\mu$ to approach the noise-limited region. 

\begin{figure}
\centering
\includegraphics[width=0.8\linewidth]{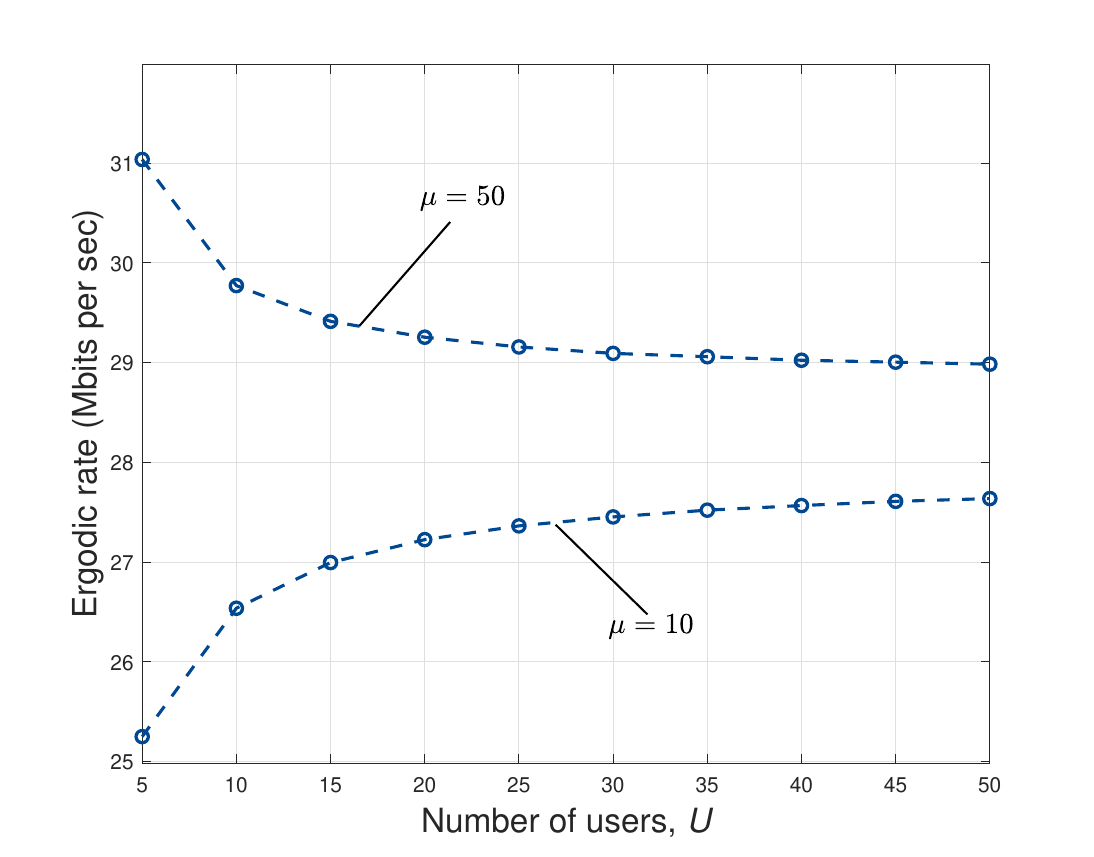}
\caption{Ergodic rate against the number of user $U$, given different port density $\mu$ when $W=3$.}\label{fig:Ce_U}
\end{figure}

The results presented in Fig.~\ref{fig:Ce_U} illustrate the ergodic rate for the satellite CUMA network under varying density $\mu$ and the number of users $U$. It is observed that for lower density ($\mu=10$), an increase in the number of users leads to a corresponding increase in the ergodic rate. Conversely, for higher density ($\mu=50$), the trend is reversed. This behavior can be attributed to the relationship between density and average SNR. In the case of lower density ($\mu=10$), the average SNR is relatively low, and the system tends to be noise-limited. As a result, increasing the number of users does not significantly degrade the SINR (hence also outage probability) of the typical user. Instead, the ergodic rate increases directly, as expressed in \eqref{Eq:Ergodic_Rate_Exactly}. In contrast, for higher density ($\mu=50$), the system becomes more susceptible to interference, and an increase in $U$ leads to a sharp decline in SINR and an increase in the outage probability of the typical user, thereby reducing the ergodic rate. Nonetheless, the higher density $\mu=50$ consistently outperforms the lower density $\mu=10$ in terms of ergodic rate, and as the number of users increases, the performance gap between the two densities begins to converge.

\begin{figure}
\centering
\includegraphics[width=0.8\linewidth]{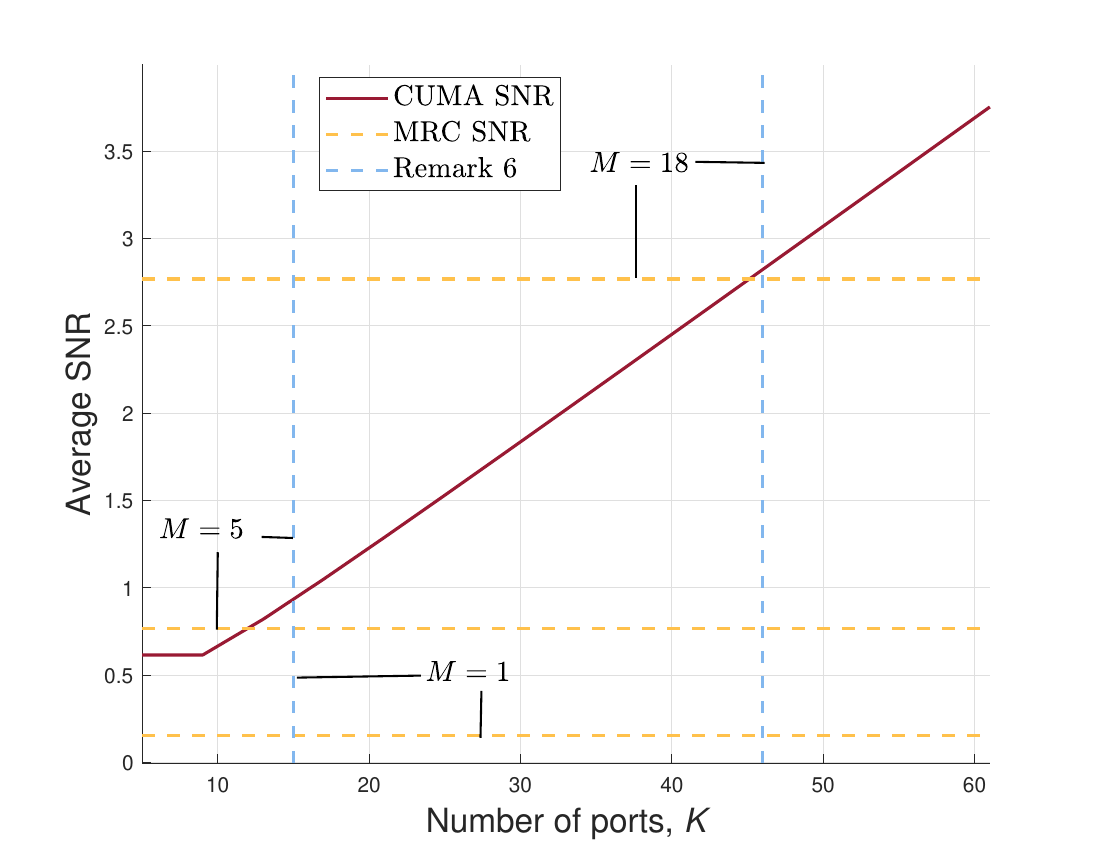}
\caption{Average SNR for different approaches against the number of ports $K$, given $W = 3$ and $\epsilon=7$. }\label{fig:Average_SNR_N}
\end{figure} 

\begin{figure}
\centering
\includegraphics[width=0.8\linewidth]{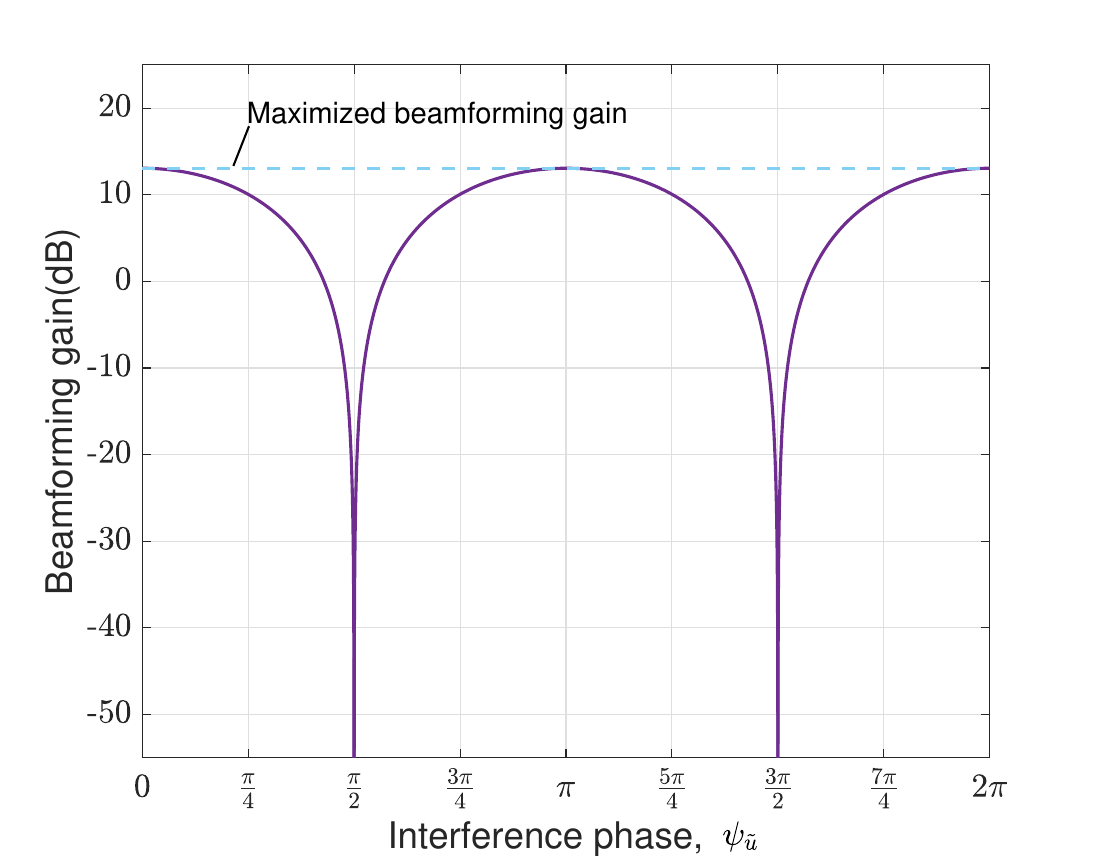}
\caption{The beamforming gain for interference against the interference phase $\psi_{\tilde{u}}$, given $\psi_{u} = \pi$, $K=51$, $W = 5$.}\label{fig:Beamforming_Gain}
\end{figure}

Now, the results in Fig.~\ref{fig:Average_SNR_N} illustrate that under a noise-limited system, the impact of the number of ports $K$ on the average SNR of CUMA. The results of using MRC are also provided as a benchmark here. The average SNR of CUMA and MRC can be obtained by neglecting the interference term in \eqref{Eq:SINR_exact} and \eqref{Eq:MRC_SINR}, respectively. Note the average SNR 
of CUMA is linearly increased with respect to $K$ under the compact fluid antenna setting, as discussed in Section \ref{Discussion_beamfroming_gain}. With $M=18$, the results in Remark \ref{Remark:CUMA_Outperform_MRC_Noise} exactly specify the number of ports $K$ required to outperform MRC. As for $M=3$, Remark \ref{Remark:CUMA_Outperform_MRC_Noise} suggests a higher number of ports, as it derived under the assumption of the compact fluid antenna. $M=1$ corresponds to a single-input single-output (SISO) scenario, which is observed to achieve the worst performance.

Finally, Fig.~\ref{fig:Beamforming_Gain} illustrates the impact of the initial phase on the beamforming gain for interference. Two angular intervals are observed: one around the initial phase of the signal, $\pi$, and the other around $0$ (which also corresponds to $2\pi$). This behavior results from the periodicity of trigonometric functions, where signals separated by the phase of $\pi$ exhibit the same power. Consequently, any interfering signal with an initial phase different from these two intervals is attenuated.

\section{Conclusion}\label{sec:conclude}

This paper explored a satellite CUMA network in which the satellite is equipped with CUMA to mitigate inter-user interference during uplink transmission. Our emphasis was on scenarios involving densely distributed users, where the azimuth and elevation angles of the target and interfering users are indistinguishable, thereby preventing MRC from differentiating between them. We derived the outage probability expression for the considered system and provided an approximate closed-form expression of it. Numerical results have demonstrated that, while $O$-CUMA outperforms $N$-CUMA with narrow bandwidth, $N$-CUMA exhibits superior performance as the bandwidth increases. This indicates that $N$-CUMA is the preferred method under broadband conditions. These findings suggest that CUMA presents a viable solution for mitigating regional interference in satellite communications scenarios, while also reducing the number of RF chains, thereby lowering costs. Further research could investigate the performance of satellite CUMA systems without assuming a densely distributed user scenario, wherein the angles of arrival may be arbitrary.

\appendix

\subsection{Proof of Lemma \ref{Lemma_Signal_Expression}}\label{appendix:Proof_Lemma_Signal_Expression}
Note $W\lambda$ denotes the length of the fluid antenna, where $\lambda$ is the propagation wavelength. Therefore, the propagation distance of the radio wave emitted by user $u$ along the fluid antenna is $W\lambda$. According to \eqref{Path_Loss} and \eqref{CUMA_signal}, the received CUMA in-phase signal magnitude $\sqrt{\alpha_{\rm I}}$ can be expressed as
\begin{align}\label{Signal_Approximation_Step1}
    \sqrt{\alpha_{\rm I}} &= \zeta_{u}^{\frac{1}{2}}\sum_{k\in {\mathcal{K}_{1}}}{\mathrm{real}}\left (e^{j\psi_{u}}e^{j2\pi\frac{k-1}{\mu}}\right) \nonumber \\
    &= \zeta_{u}^{\frac{1}{2}}\sum_{k\in {\mathcal{K}_{1}}} \cos{\left(\psi_{u}+2\pi\frac{(k-1)}{\mu}(k-1)\right)},
\end{align}
which can be seen as an aggregated discrete sampling over a length of $W\lambda$. Due to the periodicity of the cos function, \eqref{Signal_Approximation_Step1} is equivalent to sampling the cosine function over $\lambda$, scaled by $W$, as shown in \eqref{Signal_Expression}. The lower and upper bounds can be found based on the periodicity of cos function.  
\subsection{Proof of Lemma \ref{Lemma_Upper_Limit_Expression}}\label{appendix:Proof_Lemma_Upper_Limit_Expression}
The difference between $k_{\mathrm{up}}$ and $k_{\mathrm{low}}$ can be directly calculated using the results given in Lemma \ref{Lemma_Signal_Expression}. For notation simplify, let $A = \frac{5}{4}\mu-\frac{\psi_{u}}{2\pi}\mu$ and $B = \frac{3}{4}\mu-\frac{\psi_{u}}{2\pi}\mu = A- \frac{\mu}{2}$. Note both $A$ and $B$ can be written as a combination of their integer and fractional parts, given by
\begin{align}
\begin{cases}
    A &= \text{n}\left(A\right)+\text{f}\left(A\right), \\
    B &= \text{n}\left(A\right)+\text{f}\left(A\right)-\text{n}\left(\frac{\mu}{2}\right)-\text{f}\left(\frac{\mu}{2}\right),
\end{cases}
\end{align}
where $\text{n}\left(\cdot\right)$ and $\text{f}\left(\cdot\right)$ represent the operations of extracting the integer part and the fractional part, respectively.

Consequently, the difference between $k_{\mathrm{up}}$ and $k_{\mathrm{low}}$ is given by
\begin{align}
    k_{\mathrm{up}}-k_{\mathrm{low}} &= \left\lfloor A \right\rfloor - \left\lceil B \right\rceil \nonumber \\ 
    &= \text{n}\left(\frac{\mu}{2}\right)-\left\lceil \text{f}\left(A\right)-\text{f}\left(\frac{\mu}{2}\right) \right\rceil.
\end{align}
Given $\mu$ as an even integer, such that  
$\text{f}\left(\frac{\mu}{2}\right) = 0$ and $\text{n}\left(\frac{\mu}{2}\right) = \frac{\mu}{2}$. Since $\psi_{u} \in(0,2\pi)$, it follows that $A>0$, and thus $\text{f}\left(A\right) > 0$, leading to $\left\lceil \text{f}\left(A\right) \right\rceil = 1$. Finally, we obtain \eqref{Expression_k_up}.

\subsection{Proof of Theorem \ref{Theorem:Simplified_Signal_Power}}\label{appendix:Proof_Theorem_Simplified_Signal}
The received CUMA signal can be expressed as \eqref{Appendix_B_1}, as shown at the top of next page, where $(a)$ is the result of splitting the summation, $(b)$ is obtained by taking the real part of the exponential sum \cite[(3.1.10)]{Geometric_Series}, and $(c)$  involves the substitution of \eqref{Expression_k_up}. Through further simplification using trigonometric identities that $\sin{\left(x+\frac{\pi}{2}\right)} =\cos{\left(x\right)}$, $\cos{\left(x+\frac{\pi}{2}\right)} = -\sin{\left(x\right)}$, $\sin{\left(x+y\right)} =\cos{\left(x\right)}\sin{\left(y\right)}+\sin{\left(x\right)\cos{\left(y\right)}}$, and by substituting \eqref{k_Low_up_Expression} in, we reached $(d)$. 
\begin{figure*}[]
\centering
\begin{align}\label{Appendix_B_1}
\sqrt{\alpha _{\mathrm{I}}}
 &\overset{(a)}{=} \zeta_{u}^{\frac{1}{2}}W\left[\sum_{ k=2}^{k_{\mathrm{up}}}\cos{\left(\psi_{u}+ \frac{2\pi}{\mu}(k-1)\right)}-\sum_{ k=2}^{k_{\mathrm{low}}-1}\cos{\left(\psi_{u}+\frac{2\pi}{\mu}(k-1)\right)} \right] \nonumber \\
  &\overset{(b)}{=}\frac{\zeta_{u}^{\frac{1}{2}}W}{\sin{\left(\frac{\pi}{\mu}\right)}}\left[\sin{\left(\frac{\pi}{\mu}k_{\mathrm{up}}\right)}\cos{\left(\psi_{u}+\frac{\pi}{\mu}(k_{\mathrm{up}}-1)\right)}-\sin{\left(\frac{\pi}{\mu}(k_{\mathrm{low}}-1)\right)}\cos{\left(\psi_{u}+\frac{\pi}{\mu}(k_{\mathrm{low}}-2)\right)}\right] \nonumber \\
  &\overset{(c)}{=}\frac{\zeta_{u}^{\frac{1}{2}}W}{\sin{\left(\frac{\pi}{\mu}\right)}}\bigg[\sin{\left(\frac{\pi}{\mu}k_{\mathrm{low}}+\frac{\pi}{2}-\frac{\pi}{\mu}\right)}\cos{\left(\psi_{u}+\frac{\pi}{\mu}k_{\mathrm{low}}+\frac{\pi}{2}-\frac{2\pi}{\mu}\right)} -\sin{\left(\frac{\pi}{\mu}k_{\mathrm{low}}-\frac{\pi}{\mu}\right)}\cos{\left(\psi_{u}+\frac{\pi}{\mu}k_{\mathrm{low}}-\frac{2\pi}{\mu}\right)}\bigg]\nonumber \\
  & \overset{(d)}{=}-\frac{\zeta_{u}^{\frac{1}{2}}W}{\sin{\left(\frac{\pi}{\mu}\right)}}\sin{\left(\psi_{u}-\frac{\pi}{\mu}+\frac{2\pi}{\mu}\left\lceil \left(\frac{3}{4}-\frac{\psi_{u}}{2\pi}\right)\mu \right\rceil\right)}.
\end{align}
\hrulefill
\end{figure*}

Finally, set $t = \frac{3}{4}-\frac{\psi_{0}}{2\pi}$ and $V = \frac{\sin{\left(\frac{\pi}{\mu}\right)}}{W}$, the received CUMA signal magnitude $\sqrt{\alpha _{\mathrm{I}}}$ is given by 
\begin{align}
    \sqrt{\alpha _{\mathrm{I}}}=\frac{\zeta_{u}^{\frac{1}{2}}\cos{\left(-2\pi t-\frac{\pi}{\mu}+\frac{2\pi}{\mu}\left\lceil t\mu \right\rceil\right)}}{V}.
\end{align}
Subsequently, the received CUMA signal power can be expressed as \eqref{eq:Approximated_Expression_CUMA_Signal_Power}.

\subsection{Proof of Theorem \ref{Theorem_Distribution_Signal}}\label{Appendix:Proof_Theorem_Signal_Distribution}
To begin with, we set $x = -2\pi t-\frac{\pi}{\mu}+\frac{2\pi}{\mu}\left\lceil t\mu\right\rceil$ in this subsection to simplify the notation in the derivation process. We also find the following lemma useful in derivation.
\begin{lemma}\label{Lemma:continuous_RV_x}
    The random variable $x$ follows a continuous uniform distribution as $\mathcal{U}\left(-\frac{\pi}{\mu},\frac{\pi}{\mu}\right)$.
\end{lemma}
\begin{proof}
    According to Theorem \ref{Theorem:Simplified_Signal_Power}, $t = \frac{3}{4}-\frac{\psi_{0}}{2\pi}$ follows a continuous uniform random variable as $\mathcal{U}(-\frac{1}{4},\frac{3}{4})$. Therefore, $\left\lceil t\mu \right\rceil$ follows a discrete uniform distribution on $\left[\left\lceil -\frac{\mu}{4} \right\rceil,\left\lceil \frac{3\mu}{4} \right\rceil\right]$. Given $t\in\left(\frac{n}{\mu}, \frac{n+1}{\mu}\right]$, we have $\left\lceil t\mu \right\rceil = n+1$, where $n$ is a integer. Furthermore, we claim that given $t\in\left(\frac{n}{\mu}, \frac{n+1}{\mu}\right]$, $x = -2\pi t-\frac{\pi}{\mu}+\frac{2\pi}{\mu}\left\lceil t\mu\right\rceil$ follows a continuous uniform distribution as $\mathcal{U}\left(-\frac{\pi}{\mu},\frac{\pi}{\mu}\right)$, which is generalized for any integer $n$. Due to the fact that given $t \in (-\frac{1}{4},\frac{3}{4})$, the distribution of $x$ remains invariant, we conclude that for $t\in\left[-\frac{1}{4}, \frac{3}{4}\right]$, $x$ follows a continues uniform distribution as $\mathcal{U}\left(-\frac{\pi}{\mu},\frac{\pi}{\mu}\right)$, where the continuity of $x$ is ensured by $t$.
\end{proof}

The expression of $\alpha_{\mathrm{I}}$ given in \eqref{eq:Approximated_Expression_CUMA_Signal_Power} can be rewritten as
\begin{align}\label{Eq:Rewriteen_Alpha_I}
    \alpha_{\mathrm{I}} &\overset{(a)}{=} \frac{\zeta_{u}}{V^2} \cos^2{(x)} \nonumber \\
    &\overset{(b)}{=} \frac{\zeta_{u}}{2V^2} \left(1+\cos{(\left|2x\right|)}\right)\nonumber \\
    &\overset{(c)}{=} \frac{\zeta_{u}}{2V^2} \left(1+\cos{(q)}\right),
\end{align}
where $(a)$ applies the notation defined in this subsection and basic simplification, $(b)$ uses the fact that $\cos(\cdot)$ is an even function and $(c)$ is the result of notation $q = |2x|$. According to Lemma \ref{Lemma:continuous_RV_x}, $2x\in\left[-\frac{2\pi}{\mu},\frac{2\pi}{\mu}\right]$, and therefore $q \in\left[0,\frac{2\pi}{\mu}\right]$. Given $\mu$ as a positive integer, $\alpha_{\mathrm{I}}$ is now monotonic over the domain of $y$ and, therefore, invertible. The inverse function of $\alpha_{\mathrm{I}}$ with respect to $q$ is given by
\begin{align}
    q = \arccos{\left(\frac{2V^2}{\zeta_{u}}\alpha_{\mathrm{I}}-1\right)}.
\end{align}
Subsequently, the PDF of $\alpha _{\mathrm{I}}$ can be derived through the following PDF transform as 
\begin{align}
    f_{\alpha _{\mathrm{I}}}(\alpha) &= \frac{\mu}{2\pi}\left|\frac{\partial}{\partial\alpha} \arccos\left(\frac{2V^2}{\zeta_{u}}\alpha-1\right)\right| \nonumber \\
    &= \frac{\mu}{2\pi}\frac{2V^2}{\zeta_{u}}\sqrt{\frac{1}{1-\left(\frac{2V^2}{\zeta_{u}}\alpha-1\right)^2}},
\end{align}
which finally simplified as \eqref{Eq:PDF_Signal_Power}.

\subsection{Proof of Theorem \ref{Theorem:Interference_CIT}}\label{appendix:Proof_Interference_CIT}
We find the following Lemma useful in derivation. 
\begin{lemma}
     Given $\psi_{\tilde{u}}$ uniformly distributed on $(0,2\pi)$, we have the following moment expressions as, \begin{align}\label{First_Moment}\mathrm{E}\left[\sin{\left(\psi_{\tilde{u}}-\frac{\pi}{\mu}+\frac{2\pi}{\mu}\left\lceil t\mu \right\rceil\right)}\right] &= 0, \\ \label{Second_Moment}\mathrm{E}\left[\sin^2{\left(\psi_{\tilde{u}}-\frac{\pi}{\mu}+\frac{2\pi}{\mu}\left\lceil t\mu \right\rceil\right)}\right] &= \frac{1}{2}, \\ \label{Fourth_Moment}\mathrm{E}\left[\sin^4{\left(\psi_{\tilde{u}}-\frac{\pi}{\mu}+\frac{2\pi}{\mu}\left\lceil t\mu \right\rceil\right)}\right] &= \frac{3}{8}.
    \end{align}
\end{lemma}
\begin{proof}
    The first moment can be easily verified based on the periodicity of sin function. The rest of the expressions can be found through trigonometric power reduction identities. 
\end{proof}

According to $\eqref{eq:Interference_Expression}$ and with the help of \eqref{Second_Moment}, the mean of interference power from user $\tilde{u}$ is given by
\begin{align}
\mathrm{E}\left[Y_{\tilde{u}}\right] = \frac{\zeta_{\tilde{{u}}}}{V^2}\mathrm{E}\left[\sin^2{\left(\psi_{\tilde{u}}-\frac{\pi}{\mu}+\frac{2\pi}{\mu}\left\lceil t\mu\right\rceil\right)}\right]= \frac{\zeta_{\tilde{{u}}}}{2V^2}.
\end{align}
Furthermore, the variance of interference power from user $\tilde{u}$ is expressed as
\begin{align}
\mathrm{Var}\left[Y_{\tilde{u}}\right] &\overset{(a)}{=}  \frac{\zeta_{\tilde{{u}}}^2}{V^4}\mathrm{E}\left[\sin^4{\left(\psi_{\tilde{u}}-\frac{\pi}{\mu}+\frac{2\pi}{\mu}\left\lceil t\mu\right\rceil\right)}\right] \nonumber \\
&- \frac{\zeta_{\tilde{{u}}}^2}{V^4}\mathrm{E}^2\left[\sin^2{\left(\psi_{\tilde{u}}-\frac{\pi}{\mu}+\frac{2\pi}{\mu}\left\lceil t\mu\right\rceil\right)}\right] \nonumber\\
 & \overset{(b)}{=} \frac{1}{8}\frac{\zeta_{\tilde{{u}}}^2}{V^4},
\end{align}
where $(a)$ uses $\mathrm{Var}[x] = \mathrm{E}[x^2]-\mathrm{E}^2[x]$ and $(b)$ substitutes \eqref{Second_Moment} and \eqref{Fourth_Moment} in. 

In the following, we first apply the Lyapunov version of the Central Limit Theorem (CLT) \cite{CLT} to demonstrate that the interference power $\beta_{\mathrm{I}}$ is approximately Gaussian, then used truncated Gaussian to fit its definition domain. To begin with, the $(2+\delta)$-th order central moment of $Y_{\tilde{u}}$ is bounded as
\begin{align}\label{Bounded_Moments}
&{\rm{E}}\left[ {{{\left| {{Y_{\tilde u}} - {\rm{E}}\left[ {{Y_{\tilde u}}} \right]} \right|}^{2 + \delta }}} \right] = \frac{{\zeta _{\tilde u}^{2 + \delta }}}{{{V^{4 + 2\delta }}}}\\
& \times {\rm{E}}\left[ {{{\left| {{{\sin }^2}\left( {{\psi _{\tilde u}} - \frac{\pi }{\mu } + \frac{{2\pi }}{\mu }\left\lceil {t\mu } \right\rceil } \right) - \frac{1}{2}} \right|}^{2 + \delta }}} \right] \le C{\left( {\frac{{\zeta _{\tilde u}^2}}{{{V^4}}}} \right)^{1 + \frac{\delta }{2}}},\nonumber
\end{align}
where constant $C$ is the upper bound of the expression $\mathrm{E}\left[\left|\sin^2{\left(\psi_{\tilde{u}}-\frac{\pi}{\mu}+\frac{2\pi}{\mu}\left\lceil t\mu\right\rceil\right)}-\frac{1}{2} \right|^{2+\delta} \right]$. Note \eqref{Bounded_Moments} is valid for any positive $\delta$, due to the fact that $0\leq\sin^2{(x)}\leq1$. Subsequently, we have
\begin{align}
    &\quad  
    \frac{\sum_{\substack{\tilde {u}=1\\ \tilde {u}\ne u}}^{U} \mathrm{E} \left[ \left|Y_{\tilde{u}} - \mathrm{E}\left[Y_{\tilde{u}}\right] \right|^{2+\delta} \right]}{\left[\sum_{\substack{\tilde {u}=1\\ \tilde {u}\ne u}}^{U} \mathrm{Var}\left[Y_{\tilde{u}}\right]\right]^{1+\frac{\delta}{2}}} \leq \frac{\sum_{\substack{\tilde {u}=1\\ \tilde {u}\ne u}}^{U}C\left(\frac{\zeta_{\tilde{u}}^2}{V^4}\right)^{1+\frac{\delta}{2}} }{\left[\sum_{\substack{\tilde {u}=1\\ \tilde {u}\ne u}}^{U}\frac{1}{8}\frac{\zeta_{\tilde{{u}}}^2}{V^4}\right]^{1+\frac{\delta}{2}}}.
\end{align}
It is easy to show that $\left[\sum_{\substack{\tilde {u}=1\\ \tilde {u}\ne u}}^{U}\frac{1}{8}\frac{\zeta_{\tilde{{u}}}^2}{V^4}\right]^{1+\frac{\delta}{2}}$ grows as a polynomial rate in $U$, while ${\sum_{\substack{\tilde {u}=1\\ \tilde {u}\ne u}}^{U}C\left(\frac{\zeta_{\tilde{u}}^2}{V^4}\right)^{1+\frac{\delta}{2}}}$ increases linearly with $U$. Therefore, we have 
\begin{align}
\lim_{U \to \infty}  
    \frac{\sum_{\substack{\tilde {u}=1\\ \tilde {u}\ne u}}^{U} \mathrm{E} \left[ \left|Y_{\tilde{u}} - \mathrm{E}\left[Y_{\tilde{u}}\right] \right|^{2+\delta} \right]}{\left[\sum_{\substack{\tilde {u}=1\\ \tilde {u}\ne u}}^{U} \mathrm{Var}\left[Y_{\tilde{u}}\right]\right]^{1+\frac{\delta}{2}}} = 0, 
\end{align}
which satisfies the Lyapunov's condition \cite[(27.16)]{CLT}. 

Despite CLT captures the main characteristics, $\beta_{\mathrm{I}}$ represents the power of the overall interference, which is strictly positive. In contrast, the Gaussian distribution is defined over the entire set of real numbers, including negative values. Therefore, we adopt the truncated Gaussian distribution for $\beta_{\mathrm{I}}$ in the following. 
According to \cite[p.20]{Truncated_Gaussian}, we have the distribution of $\beta_{\mathrm{I}}$, that is expressed in a truncated Gaussian form, with a PDF given by
\begin{align}
    f_{\beta_{\mathrm{I}}}(\beta) = \frac{1}{1-\Phi{\left(-\frac{\omega}{\kappa}\right)}}\frac{1}{\sqrt{2\pi\kappa^2}}e^{-\frac{\left(\beta-\omega\right)^2}{2\kappa^2}},
\end{align}
where $\omega = \frac{1}{2 V^2}\sum_{\substack{\tilde {u}=1\\ \tilde {u}\ne u}}^{U}\zeta_{\tilde{u}}$ and $\kappa = \sqrt{\frac{1}{8 V^4}\sum_{\substack{\tilde {u}=1\\ \tilde {u}\ne u}}^{U}\zeta_{\tilde{u}}^2}$ are the mean and standard deviation of $\beta_{\mathrm{I}}$ before truncation, respectively. $\Phi(\cdot)$ is the CDF of standard Gaussian distribution. Given $\Phi(-x) = 1-\Phi(x)$, we have the PDF expressed as \eqref{Eq: Interference_PDF}. 

\subsection{Proof of Corollary \ref{Coro:OP_Single_Integral}}\label{Appendix:Proof_OP_Single_Integral}
The outage probability of CUMA by using only in-phase component is given by \eqref{OP_Single_Integral_Derivation}, as shown at the top of next page, where $(a)$ comes from the definition of outage probability, $(b)$ substitutes \eqref{Eq:PDF_Signal_Power} and \eqref{Eq:PDF_Interference_Plus_Noise} in. $(c)$ evaluates the integral over $\beta$. After further simplicity, we reached \eqref{Eq:OP_Single_Integral} and completed the proof.

\begin{figure*}[!htbp]
\begin{align}\label{OP_Single_Integral_Derivation}
        {\mathcal O}_{\mathrm{I}}\left(\gamma\right) &\overset{(a)}{=} 1-\mathbb{P}\left\{\tilde{\beta} \leq \frac{\alpha_{\mathrm{I}}}{\gamma}\right\} \nonumber \\
        &\overset{(b)}{=}1-\int_{\alpha=\frac{\cos^2{\left(\frac{\pi}{\mu}\right)}\zeta_{u}}{V^2}}^{\frac{\zeta_{u}}{V^2}}\frac{\mu}{2\pi}\sqrt{\frac{V^2}{\zeta_{u}\alpha-V^2\alpha^2}}\int_{\tilde{\beta}=\frac{\bar{K}}{2\Gamma}}^{\frac{\alpha}{\gamma}}\frac{1}{\Phi{\left(\frac{\omega}{\kappa}\right)}}\frac{1}{\sqrt{2\pi\kappa^2}}e^{-\frac{\left(\tilde{\beta}-\omega-\frac{\bar{K}}{2\Gamma}\right)^2}{2\kappa^2}}d\tilde{\beta}d\alpha \nonumber \\
    &\overset{(c)}{=}1-\int_{\alpha=\frac{\cos^2{\left(\frac{\pi}{\mu}\right)}\zeta_{u}}{V^2}}^{\frac{\zeta_{u}}{V^2}}\frac{\mu}{2\pi}\sqrt{\frac{V^2}{\zeta_{u}-V^2\alpha^2}}\frac{1}{\Phi{\left(\frac{\omega}{\kappa}\right)}}\left[\Phi{\left(\frac{\frac{\alpha}{\gamma}-\omega-\frac{\bar{K}}{2\Gamma}}{\kappa}\right)}- \Phi{\left(-\frac{\omega}{\kappa}\right)} \right]d\alpha.
    \end{align}
    \hrulefill
\end{figure*}

\subsection{Proof of Theorem \ref{Theorem:PDF_Z}}\label{Appendix:Proof_PDF_Z}
The PDF of $Z_{\mathrm{I}} = \frac{\alpha_{\mathrm{I}}}{\tilde{\beta_{\mathrm{I}}}}$ can be expressed as
\begin{align}\label{PDF_Derivation_Z}
    f_{Z_{\mathrm{I}}}(z) &\overset{(a)}{=}\int_{-\infty}^{\infty}\left|{\tilde{\beta}}\right|f_{\alpha_{\mathrm{I}},\tilde{\beta_{\mathrm{I}}}}(z\tilde{\beta},\tilde{\beta})d\tilde{\beta} \nonumber\\
    &\overset{(b)}{=}\int_{0}^{\infty}{\tilde{\beta}}f_{\alpha_{\mathrm{I}},\tilde{\beta_{\mathrm{I}}}}(z\tilde{\beta},\tilde{\beta})d\tilde{\beta} \nonumber\\
    &\overset{(c)}{=}\int_{0}^{\infty}{\tilde{\beta}}f_{\alpha_{\mathrm{I}}}(z\tilde{\beta})f_{\tilde{\beta_{\mathrm{I}}}}(\tilde{\beta})d\tilde{\beta},
\end{align}
where $(a)$ is the result of the definition for joint PDF, $(b)$ follows the fact that $\tilde{\beta}$ is positive, and $(c)$ is due to the independence between $\alpha_{\mathrm{I}}$ and $\tilde{\beta_{\mathrm{I}}}$ that ensured by Remark \ref{Remark:independence}. Subsequently, accroding to \eqref{PDF_Derivation_Z}, with the PDF of $\alpha_{\mathrm{I}}$ given in Theorem \ref{Theorem_Distribution_Signal} and the PDF of $\tilde{\beta_{\mathrm{I}}}$ given in Corollary \ref{Coro:PDF_Interference_Plus_Noise},
we have the PDF expression of $Z_{\mathrm{I}}$, that expressed as
\begin{multline}
    f_{Z_{\mathrm{I}}}(z) = \int_{\frac{\cos^2{(\frac{\pi}{\mu})}\zeta_{u}}{z V^2}}^{\frac{\zeta_{u}}{z V^2}}{\tilde{\beta}}\frac{\mu}{2\pi}\sqrt{\frac{V^2}{\zeta_{u}z\tilde{\beta}-V^2z^2\tilde{\beta}^2}}  \\
    \times\frac{1}{\Phi{\left(\frac{\omega}{\kappa}\right)}}\frac{1}{\sqrt{2\pi\kappa^2}}e^{-\frac{\left(\tilde{\beta}-\omega-\frac{\bar{K}}{2\Gamma}\right)^2}{2\kappa}}d\tilde{\beta}.
\end{multline}
After further simplification, we have the PDF of $Z_{\mathrm{I}}$ that expressed as \eqref{Eq:PDF_Z}.

\subsection{Proof of Theorem \ref{Theorm:Set_K_2}}\label{Appendix:Theorm_Set_K_2}
The received CUMA signal power using $\mathcal{K}_{2}$ can be expressed as
\begin{multline}
\alpha_{\mathrm{I},\mathcal{K}_{2}} \overset{(a)}{=} 
\bigg|\sum_{k=2}^{K} 
\cos\left(\psi_{u} + \frac{2\pi}{\mu}(k-1)\right) \\
- \sum_{k \in {\mathcal{K}}_{1}} 
\cos\left(\psi_{u} + \frac{2\pi}{\mu}(k-1)\right)\bigg|^2 \\
\overset{(b)}{=}  \left|\frac{\sin\left(\frac{\pi}{\mu}K\right)}{\sin{\left(\frac{\pi}{\mu}\right)}}\cos\left(\psi_{u}+W\pi\right)-\sqrt{\alpha_{\mathrm{I},\mathcal{K}_{1}}}\right|^2,
\end{multline}
where $(a)$ comes from the fact that $\mathcal{K}_{1} \cup \mathcal{K}_{2} $ contains all the ports of the fluid antenna, $(b)$ is derived from \cite[(3.1.10)]{Geometric_Series} and uses result in Theorem \ref{Theorem:Simplified_Signal_Power}. Given a sufficiently large $K$, we have $\sin{\left(\frac{\pi}{\mu}K\right)} = \sin{\left(W\pi\frac{K}{K-1}\right)} \approx \sin{\left(W\pi\right)} = 0$. As a result, we claim that given a sufficiently large $K$, $\alpha_{\mathrm{I},\mathcal{K}_{2}} \approx \alpha_{\mathrm{I},\mathcal{K}_{1}}$. 

The proof on interference can be derived through a similar derivation process and is omitted here due to space limit.

\subsection{Proof of Theorem \ref{Theorem_CDF_Dif}}\label{Proof_Theorem_CDF_Dif}
In the following, we present the derivation of the difference in the CDF for the interval $\frac{\cos^2\left(\frac{\pi}{\mu}\right)\zeta_u}{V^2} < Y <\frac{\zeta_u}{V^2}$. The derivation of the CDF difference for the range $0<Y\leq\frac{\cos^2\left(\frac{\pi}{\mu}\right)\zeta_u}{V^2}$
is straightforward and thus omitted for brevity. We first note that the CDF expression of both $F_{\alpha_{\mathrm{I}}}(Y)$ and $F_{Y_{\tilde{u}}}(Y)$ can be rewritten in terms of the integration over the PDF of $Y_{\tilde{u}}$, $f_{Y_{\tilde{u}}}(y)$, given by
\begin{equation}\label{CDF_in_PDF}
    \left\{\begin{aligned}
    F_{\alpha_{\mathrm{I}}}(Y) &= \frac{\mu}{2}\int_{\frac{\cos^2{\left(\frac{\pi}{\mu}\right)}\zeta_{u}}{V^2}}^{Y}f_{Y_{\tilde{u}}}(y)dy,\\
    F_{Y_{\tilde{u}}}(Y)&= \int_{0}^{Y}f_{Y_{\tilde{u}}}(y)dy.
    \end{aligned} \right.
\end{equation}
Subsequently, we have 
\begin{align}\label{integration_equality}
    \int_{0}^{\frac{\zeta_{u}}{V^2}}f_{Y_{\tilde{u}}}(y)dy  &\overset{(a)}{=} \frac{\mu}{2}\int_{\frac{\cos^2{\left(\frac{\pi}{\mu}\right)}\zeta_{u}}{V^2}}^{\frac{\zeta_{u}}{V^2}}f_{Y_{\tilde{u}}}(y)dy,\nonumber \\
    \int_{0}^{\frac{\cos^2{\left(\frac{\pi}{\mu}\right)}\zeta_{u}}{V^2}}f_{Y_{\tilde{u}}}(y)dy &\overset{(b)}{=} \left(\frac{\mu}{2}-1\right)\int_{\frac{\cos^2{\left(\frac{\pi}{\mu}\right)}\zeta_{u}}{V^2}}^{\frac{\zeta_{u}}{V^2}}f_{Y_{\tilde{u}}}(y)dy,
\end{align}
where $(a)$ used the fact that $F_{\alpha_{\mathrm{I}}}(\frac{\zeta_{u}}{V^2}) = F_{Y_{\tilde{u}}}(\frac{\zeta_{u}}{V^2}) = 1$, and $(b)$ is the result of the additivity property of the integration. Consequently, the difference between $F_{\alpha_{\mathrm{I}}}(\cdot)$ and $F_{Y_{\tilde{u}}}(\cdot)$ can be expressed as
\begin{align}
    &F_{Y_{\tilde{u}}}(Y)-F_{\alpha_{\mathrm{I}}}(Y) \nonumber\\&\overset{(a)}{=} \int_{0}^{Y}f_{Y_{\tilde{u}}}(y)dy-\frac{\mu}{2}\int_{\frac{\cos^2{\left(\frac{\pi}{\mu}\right)}\zeta_{u}}{V^2}}^{Y}f_{Y_{\tilde{u}}}(y)dy, \nonumber \\
    &\overset{(b)}{=}\int_{0}^{\frac{\cos^2{\left(\frac{\pi}{\mu}\right)}\zeta_{u}}{V^2}}f_{Y_{\tilde{u}}}(y)dy-\left(\frac{\mu}{2}-1\right)\int_{\frac{\cos^2{\left(\frac{\pi}{\mu}\right)}\zeta_{u}}{V^2}}^{Y}f_{Y_{\tilde{u}}}(y)dy, \nonumber \\
    &\overset{(c)}{=}\left(\frac{\mu}{2}-1\right)\int_{\frac{\cos^2{\left(\frac{\pi}{\mu}\right)}\zeta_{u}}{V^2}}^{\frac{\zeta_{u}}{V^2}}f_{Y_{\tilde{u}}}(y)dy \nonumber \\&\quad\quad\quad\quad\quad\quad\quad-\left(\frac{\mu}{2}-1\right)\int_{\frac{\cos^2{\left(\frac{\pi}{\mu}\right)}\zeta_{u}}{V^2}}^{Y}f_{Y_{\tilde{u}}}(y)dy,
\end{align}
where $(a)$ is the result of \eqref{CDF_in_PDF}, $(c)$ uses the additivity property of integration, and $(c)$ substitutes \eqref{integration_equality} in. After further simplification, we have the expression given in \eqref{eq_CDF_Dif}. This completes the proof.

\bibliography{Reference_List}
\bibliographystyle{ieeetr}

\end{document}